\documentclass{amsart}

\usepackage[english]{babel}
\usepackage[utf8]{inputenc}
\usepackage[T1]{fontenc}

\usepackage{amsmath}
\usepackage{amssymb}
\usepackage{mathtools}
\usepackage{amsthm}
\usepackage{physics}
\usepackage{enumitem}

\usepackage{booktabs}
\usepackage{rotating}
\usepackage{tabularx}
\usepackage{multirow}

\usepackage{siunitx}

\usepackage{graphicx}
\usepackage{subfig}
\usepackage{wrapfig}
\usepackage{xcolor}
\usepackage{subfig}
\usepackage{tikz}

\usepackage{etoolbox} 
\AtBeginEnvironment{pmatrix}{\everymath{\displaystyle}}
\AtBeginEnvironment{bmatrix}{\everymath{\displaystyle}}

\DeclareMathOperator{\arctanh}{arctanh}

\DeclareMathOperator{\Id}{Id}
\DeclareMathOperator{\ud}{d}

\newcommand{\Z}{\mathbb{Z}}

\newcommand{\R}{\mathbb{R}}

\newcommand{\dP}{\mathrm{d}P}
\newcommand{\lequiv}{\stackrel{\text{t.d.}}{\equiv}}

\renewcommand{\vec}[1]{\mathbf{#1}}
\renewcommand{\tilde}{\widetilde}
\renewcommand{\imath}{\mathrm{i}}

\newcommand\lagrangeprime[1]{^{%
\ifcase#1 \or\prime\or\prime\prime\or\prime\prime\prime\else\mathrm{\romannumeral #1}\fi}}

\numberwithin{equation}{section}

\allowdisplaybreaks
\makeatletter 
\renewcommand{\pdv}[2]{\begingroup 
\@tempswafalse\toks@={}\count@=\z@ 
\@for\next:=#2\do 
{\expandafter\check@var\next\@nil
 \advance\count@\der@exp 
 \if@tempswa 
   \toks@=\expandafter{\the\toks@\,}%
 \else 
   \@tempswatrue 
 \fi 
 \toks@=\expandafter{\the\expandafter\toks@\expandafter\partial\der@var}}%
\frac{\partial\ifnum\count@=\@ne\else^{\number\count@}\fi#1}{\the\toks@}%
\endgroup} 
\def\check@var{\@ifstar{\mult@var}{\one@var}} 
\def\mult@var#1#2\@nil{\def\der@var{#2^{#1}}\def\der@exp{#1}} 
\def\one@var#1\@nil{\def\der@var{#1}\chardef\der@exp\@ne} 
\makeatother

\theoremstyle{plain}
\newtheorem{theorem}{Theorem}

\newtheorem{corollary}[theorem]{Corollary}
\newtheorem{lemma}[theorem]{Lemma}
\theoremstyle{definition}

\theoremstyle{remark}
\newtheorem{remark}{Remark}
\theoremstyle{remark}
\newtheorem{example}{Example}
\theoremstyle{definition}
\newtheorem*{problem*}{Problem}
\theoremstyle{definition}
\newtheorem*{conjecture*}{Conjecture}

\usepackage[defernumbers=true, style=numeric-comp,sorting=nyt,backend=biber,maxbibnames=99]{biblatex}
\begin{document}
\title[Lagrangians and integrability for add. 4th-order difference eqs]{Lagrangians and integrability for additive fourth-order difference equations}

\author{Giorgio Gubbiotti}
\address{School  of  Mathematics  and  Statistics  F07,  The  University  of  Sydney,  NSW  2006, Australia}
\email{giorgio.gubbiotti@sydney.edu.au}

\date{\today}

\begin{abstract}
    We use a recently found method to characterise all the
    invertible fourth-order difference equations linear in the
    extremal values based on the existence of a discrete Lagrangian.
    We also give some result on the integrability properties of
    the obtained family and we put it in relation with known
    classifications.
    Finally, we discuss the continuum limits of the integrable cases.
\end{abstract}

\maketitle

\section{Introduction}

Discrete equations attracted the interest of many scientists
during the past decades for several reason, 
which span from philosophical to purely practical.
For instance, several modern theory of physics led to hypothesis 
that the nature of space-time itself at very small scales, the so-called
Planck length and Planck time, is discrete.
From this assumption it follows  that discrete
systems are actually at the very foundation of physical sciences,
see \cite{Hagar2014book} for a complete discussion and perspective
on this subject.
On the other hand, discrete systems often appears in applied
sciences as tools to investigate numerically equations whose
closed form solution is not available.
In particular, discrete equations are related to finite
difference methods for solving ordinary and partial differential
equations \cite{NumericalRecipes3rd}.
All these considerations greatly stimulated the theoretical study 
of discrete systems from different points of view and perspective,
see \cite{HietarintaJoshiNijhoff2016,Elaydi2005}.

In this paper we will deal \emph{fourth-order difference equations}, that
is, functional equations for an unknown sequence $\left\{ x_{n} \right\}$
where the $x_{n+2}$ element is expressible in term of
the previous $x_{n+i}$, $i=-2,\dots,1$.
That is a fourth-order difference equation is a relation of the form:
\begin{equation}
    x_{n+2} = F\left( x_{n+1},x_{n},x_{n-1},x_{n-2} \right).
    \label{eq:rec4}
\end{equation}
Such kind of functional equations are also called \emph{recurrence relations}
of order four.
A fourth-order difference equations is called \emph{invertible} if it is
possible to solve equation \eqref{eq:rec4} in a unique way with respect
to $x_{n-2}$.
\begin{equation}
    x_{n-2} = \widetilde{F}\left( x_{n+2},x_{n+1},x_{n},x_{n-1} \right).
    \label{eq:rec4inv}
\end{equation}

In this paper we study fourth-order difference equations
which are \emph{variational}, that it they arise as extremal values
of a \emph{variational principle}.
Variational principles are one of the most powerful tools in Mathematical Physics
since Euler and Lagrange,.
The branch of mathematics studying variational principles is called
\emph{calculus of variations} and it played a fundamental r\^ole in the
development of  of theoretical mechanics \cite{Whittaker,Gold2002,LandauMech}.
Beside this, the variational principles are ubiquitous in mathematics.
To name a few, variational principles are of great help in 
in the solution of isoperimetrical problems and in the study of minimal
surfaces.
For instance, Field medalist J. Douglass won his prize for his
seminal work on minimal surfaces \cite{Douglas1931}.
For a more complete outlook on the calculus of variations, its scopes
and its applications we refer to the standard textbook on the
subject \cite{GelfandFomin1963book}.

In a recent paper \cite{Gubbiotti_dcov} we solved the 
\emph{inverse problem of calculus of variations} for difference
equations of arbitrary order $2k$ with $k>1$.
That is, in \cite{Gubbiotti_dcov} we gave a list of conditions,
expressed by a system of linear partial differential equations,
which allow us to determine whether or not a given difference
equations of arbitrary order $2k$ with $k>1$ arises from a 
variational principle.

In this paper we will use the conditions derived in \cite{Gubbiotti_dcov}
to construct the most general \emph{additive fourth-order difference equation}
admitting a Lagrangian.
An additive fourth-order difference equation is a difference equation of
the form \eqref{eq:rec4} such that it is linear also in $x_{n-2}$, that is
it has the following form:
\begin{equation}
    x_{n+2} = f\left( x_{n+1},x_{n},x_{n-1} \right) x_{n-2}+h\left( x_{n+1},x_{n},x_{n-1} \right)
    \label{eq:rec4add}
\end{equation}
Additive difference equations are trivially invertible, with inverse given
by:
\begin{equation}
    x_{n-2} = \frac{ x_{n+2}-h\left( x_{n+1},x_{n},x_{n-1} \right)}{%
    f\left( x_{n+1},x_{n},x_{n-1} \right)}.
    \label{eq:rec4addinv}
\end{equation}
The interest in this kind of equations lies in the fact that they are
the most natural generalisation of additive second-order differential
equations:
\begin{equation}
    x_{n+1} + x_{n-1} = f\left( x_{n} \right),
    \label{eq:add2}
\end{equation}
a well-known class of difference equations, including very famous
examples like the McMillan equation \cite{McMillan1971} and the additive
QRT maps \cite{QRT1988,QRT1989}.
More recently several examples of equations in the form \eqref{eq:rec4}
appeared in \cite{JoshiViallet2017,GJTV_class,GJTV_sanya}.
Especially in \cite{GJTV_class} was given a classification
of fourth-order difference equations based on the existence of two
independent invariant within a given class.
It turned out that all the examples presented in \cite{GJTV_class}
fall in the class \eqref{eq:rec4add}, and the variational structure
played an important r\^ole in understanding the regularity properties
of these examples.

In this paper, we will present the most general equations of the form \eqref{eq:rec4add} 
admitting a Lagrangian. 
Then, we will discuss the regularity properties of a subclass,
in order to compare our results with previously known results \cite{GJTV_class}.
We will produce a very general family of \emph{Liouville integrable}
fourth-order difference equations.
Using \emph{linear transformations} this family of equations will be 
split in five inequivalent \emph{canonical forms}.
We will identify these canonical forms with known equations,
reinterpreting the results of \cite{GJTV_class} from
the point of view of variational structures.
We underline that for all the canonical forms the existence of a Lagrangian 
is the key element in proving Liouville integrability.

The plan of the paper is following: in Section \ref{sec:dlagr} we introduce
the basis of the discrete calculus of variations and integrability for 
scalar difference equations \cite{Logan1973,Veselov1991,Bruschietal1991}.
In particular we will present in Theorem \ref{thm:elexist} a particular 
case of the general results given in \cite{Gubbiotti_dcov}
which is a necessary and sufficient condition for the existence
of a Lagrangian for fourth-order invertible difference equations.
Moreover, we will discuss the crucial relationship between Lagrangian
structures and integrability 
\cite{Maeda1987,Veselov1991,Bruschietal1991,TranvanderKampQusipel2016}.
In Section \ref{sec:main} we present our main result in Theorem \ref{thm:structure}.
Theorem \ref{thm:structure} completely characterises
fourth-order additive difference equations \eqref{eq:rec4add} admitting 
a Lagrangian.
Then, we present and algorithmic test to find the Lagrangian of
an additive fourth-order different equation derived from Theorem 
\ref{thm:structure} and some examples.
In section \ref{sec:int} we present a subclass of variational
equations depending on seven parameters which possess two invariants.
We discuss how to split this general family to five canonical forms
depending on three essential parameters each and prove their Liouville integrability
using the Lagrangian structure.
Our results are summarised in Theorem \ref{thm:intmultinlin}.
Then in section \ref{sec:contlim} we present the continuum limits
of the Liouville integrable canonical equations, their Lagrangian
and invariants (first integrals).
Finally in Section \ref{sec:concl} we give some conclusions
and outlook.

\section{Background material}
\label{sec:dlagr}

In this section we introduce the basic notions of Lagrangians
for difference equations of even order $2k$, that is a 
functional equation for an unknown sequence $\left\{ x_{n} \right\}$
of the following form:
\begin{equation}
    x_{n+k} = F\left( x_{n+k-1},x_{n+k-2},\dots,x_{n-k} \right), k\geq 1,
    \label{eq:recgen}
\end{equation}
and their integrability properties.
The Lagrangian formulation for difference equation was discussed
already in \cite{Logan1973}, while later accounts of this theory 
can be found in \cite{Veselov1991,Bruschietal1991} and more recently
in \cite{TranvanderKampQusipel2016,HydonMansfield2004}.
In \cite{Gubbiotti_dcov} was introduced an algorithmic method to prove whether or
not a given even-order difference equation \eqref{eq:recgen} with $k>1$.
We are going to present the main result of that paper
in Theorem \ref{thm:elexist} in our case of interest, that is the
case of invertible fourth-order difference equations.
We will then discuss the notion of Liouville integrability for
difference equations \cite{Bruschietal1991,Maeda1987,Veselov1991,CapelSahadevan2001}.
As this is not the unique definition of integrability for discrete
systems we remark that a broader discussion on integrability for discrete systems
can be found in \cite{HietarintaBook,GubbiottiASIDE16,HietarintaJoshiNijhoff2016}.

We refer the interested reader to the cited papers and books 
and references therein for a complete overview on the topics.

\subsection{Discrete Lagrangians}


A discrete action of order $k$ is a linear functional of the form:
\begin{equation}
    S\left[ x_{n} \right] = \sum_{n\in\Z} L_{n}\left( x_{n+k},x_{n+k-1},\dots,x_{n} \right).
    \label{eq:daction}
\end{equation}
The summand function 
\begin{equation}
    L_{n} = L_{n}\left( x_{n+k},x_{n+k-1},\dots,x_{n} \right)
    \label{eq:lagrgen}
\end{equation}
is called a \emph{discrete Lagrangian}.
We define an \emph{admissible variation} to be the sequence
\begin{equation}
    x_{n}\left( \varepsilon \right) = x_{n} + \varepsilon h_{n},
    \label{eq:var}
\end{equation}
where $x_{n}$ is an extremal point of the discrete
action \eqref{eq:daction} and $h_{n}$ are well-behaved
sequences as $\abs{n}\to\infty$.
The condition of having an extremal point is then given
by Rolle's theorem and is that:
\begin{equation}
    \left.\dv{S\left[ x_{n}\left( \varepsilon \right) \right]}{\varepsilon}\right|_{\varepsilon=0}
    \equiv 0.
    \label{eq:stat}
\end{equation}
Working out the condition \eqref{eq:stat}
we obtain that the extremal points of the discrete action
\eqref{eq:daction} must satisfy the following difference
equation of order $2k$:
\begin{equation}
    \sum_{l=0}^{k} \pdv{L_{n-l}}{x_{n}}\left( x_{n+k-l},x_{n+k-1-l},\dots,x_{n-l} \right) = 0.
    \label{eq:elgen}
\end{equation}
This equation is known as the 
\emph{discrete Euler--Lagrange equation}.

\begin{remark}
    We underline that in formula \eqref{eq:lagrgen}
    we allow the discrete Lagrangian to depend explicitly on $n$.
    Indeed, \emph{an autonomous difference equation \eqref{eq:recgen} can arise
    even from non-autonomous Lagrangian}.
    A simple example of this occurrence is the following:
    \begin{equation}
        L_{n} = \lambda^{-n} \left( x_{n}x_{n+1}+\frac{\kappa}{2}x_{n}^{2} \right).
        \label{eq:Ldampedho}
    \end{equation}
    The Euler--Lagrange equation \eqref{eq:elgen} of \eqref{eq:Ldampedho}
    is:
    \begin{equation}
        x_{n+1} + \kappa x_{n} +\lambda x_{n-1} = 0.
        \label{eq:dampedho}
    \end{equation}
    \label{rem:nonautL}
    Clearly equation \eqref{eq:dampedho} is autonomous, while the
    discrete Lagrangian \eqref{eq:Ldampedho} is not.
    In general, given an autonomous Lagrangian $L$ the non-autonomous
    Lagrangian 
    \begin{equation}
        L_{n} = \lambda^{-n} L\left( x_{n+k},\dots,x_{n} \right).
        \label{eq:Lnagen}
    \end{equation}
    always give raise to autonomous Euler--Lagrange equations.
    \label{rem:nautlagr}
\end{remark}

The left hand side of the discrete Euler--Lagrange equations \eqref{eq:elgen}
is sometimes called the \emph{variational derivative} of the
action \eqref{eq:daction} and denoted by $\delta S/\delta x_{n}$.
A discrete Lagrangian is called 
\emph{normal} if
\begin{equation}
    \pdv{L_{n}}{x_{n},x_{n+k}} \neq 0.
    \label{eq:lagrnormal}
\end{equation}
The discrete Euler--Lagrange equation of a normal discrete Lagrangian are of order
$2k$ whereas a non-normal discrete Lagrangian can give 
rise to discrete Euler--Lagrange equations of order at most $2k-2$
\cite{Bruschietal1991,TranvanderKampQusipel2016,Gubbiotti_dcov}.
For this reason non-normal discrete Lagrangians are degenerate and, from now on, 
we will always consider to deal with normal discrete Lagrangians.

If two discrete Lagrangians $L_{n,1}$ and $L_{n,2}$ differ by a total 
difference, i.e. there exists a function $f_{n}=f_{n}\left( x_{n+k-1},\dots,x_{n} \right)$
such that:
\begin{equation}
    L_{n,2} = L_{n,1} + \left( T_{n}-\Id \right)f_{n}\left( x_{n+k-1},\dots,x_{n} \right),
    \label{eq:lagrdiff}
\end{equation}
then they define the same discrete Euler--Lagrange equations.
This result can either be proved directly, or put into the wider
context of variational complexes.
We refer to \cite{Gubbiotti_dcov} and \cite{HydonMansfield2004,Kupershmidt1985book}
respectively for a discussion of this approaches. 
So, we can introduce the notion of equivalence on discrete Lagrangians 
as follows: two discrete Lagrangians $L_{n,1}$ and $L_{n,2}$ are called
equivalent, denoted by $\lequiv$, if they differ for a
total difference, i.e.:
\begin{equation}
    L_{1} \lequiv L_{2} \iff 
    L_{n,2} = L_{n,1} + \left( T_{n} -\Id\right) f\left( x_{n+k-1},\dots,x_{n} \right).
    \label{eq:lequiv}
\end{equation}
The relation $\lequiv$ is an equivalence relation.
That is, it possesses the following properties:
\label{prop:lequiv}
\begin{description}
    \item[Reflexivity] $L_{n}\lequiv L_{n}$.
    \item[Symmetry] If $L_{n,1}\lequiv L_{n,2}$ then $L_{n,2}\lequiv L_{n,1}$.
    \item[Transitivity] If $L_{n,1}\lequiv L_{n,2}$ and $L_{n,2}\lequiv L_{n,3}$
        then $L_{n,1}\lequiv L_{n,3}$.
\end{description}
From the above observation we obtain that equivalent Lagrangians
give raise to the same Euler--Lagrange equations \eqref{eq:elgen}.

The existence of the equivalence relation \eqref{eq:lequiv} can be interpreted
by saying that Lagrangians are not functions, but rather they are 
\emph{equivalence classes of functions}.
This fact is useful in practical application, as often choosing
properly the representative helps in simplifying the computations.

Finally, say that a discrete Lagrangian \eqref{eq:lagrgen} is a discrete Lagrangian 
for the difference equation \eqref{eq:recgen} if its discrete Euler--Lagrange equations
\eqref{eq:elgen} coincide with \eqref{eq:recgen}.
We will say that a difference equation admitting
a Lagrangian is \emph{variational}.

We now state a theorem which gives us a necessary and sufficient condition
for the existence of discrete Lagrangian in the case $k=2$.

\begin{theorem}[Gubbiotti \cite{Gubbiotti_dcov}]
    Let us assume we are given an invertible fourth-order difference 
    equation represented by a pair of equations of the form
    \eqref{eq:rec4} and \eqref{eq:rec4inv}.
    Then such pair of equations admits a Lagrangian \eqref{eq:lagrgen}
    if and only if the following partial difference equations are
    satisfied:
    \begin{subequations}
        \begin{align}
            \pdv{}{ x_{n-2}}
            \left\{
                \left(\pdv{F}{x_{n-2}}\right)^{-1}\vec{A}^{+}
                \left[\pdv{L_{n-2}}{x_{n}}\left( x_{n},x_{n-1},x_{n-2} \right)\right]
            \right\}
            =0,
            \label{eq:lagrcond}
            \\
            \pdv{}{ x_{n+2}}
            \left\{
                \left(\pdv{\widetilde{F}}{x_{n+2}}\right)^{-1}\vec{A}^{-}
                \left[\pdv{L_{n}}{x_{n}}\left( x_{n+2},x_{n+1},x_{n} \right)\right]
            \right\}
            =0,
            \label{eq:lagrcondinv}
        \end{align}
        \label{eq:lagrcondfb}
    \end{subequations}
    where:
    \begin{subequations}
        \begin{align}
        \vec{A}^{+} &= \pdv{F}{x_{n-2}}\pdv{}{x_{n-1}}- 
        \pdv{F}{x_{n-1}}\pdv{}{x_{n-2}},
        \label{eq:annihil}
        \\
        \vec{A}^{-} &= \pdv{\widetilde{F}}{x_{n+2}}\pdv{}{x_{n+1}}- 
        \pdv{\widetilde{F}}{x_{n+1}}\pdv{}{x_{n+2}},
        \label{eq:annihil2}
        \end{align}
        \label{eq:annihilop}
    \end{subequations}
    are two linear differential operators called 
    \emph{forward annihilation operator} and \emph{backward annihilation operator} 
    respectively.
    \label{thm:elexist}
\end{theorem}

\begin{remark} 
    The forward annihilation operator defined by equation \eqref{eq:annihil}
    have this name because \emph{for every} functions of the form
    \begin{equation}
        G=G\left(F\left( x_{n+1},x_{n},x_{n-1},x_{n-2} \right) ,x_{n+1},x_{n} \right).
        \label{eq:Gdeff}
    \end{equation}
    we have $\vec{A}^{+}\left( G \right) \equiv 0$.
    In the same way the backward annihilation operator \eqref{eq:annihil2}
    has this name because \emph{for every} functions of the form
    \begin{equation}
        \tilde{G}=
        \tilde{G}\left(x_{n} ,x_{n-1},\tilde{F}\left( x_{n+2},x_{n+1},x_{n},x_{n-1} \right) \right).
        \label{eq:Gdefb}
    \end{equation}
    we have $\vec{A}^{-}\left( G \right) \equiv 0$.
    In \cite{Gubbiotti_dcov} it was proved that the operators 
    \eqref{eq:annihilop} are the most general linear differential 
    operators with such properties.
    Every other linear differential operators with the same properties 
    are their multiple.
    In \cite{Gubbiotti_dcov} it was also noted that the annihilation
    operators \eqref{eq:annihilop} are the one-dimensional analog of 
    the operators $Y^{l}$ and $Z^{-l}$, for $l\in\Z$, defined in 
    \cite{LeviYamilov2011,Garifullin2012,GarifullinYamilov2012}.
    These operators annihilates all the dependent shifts of
    a quad equation, while $\vec{A}^{\pm}$ annihilates the dependent
    shifts of a scalar difference equation.
    Moreover, these operators have application also in the theory of Darboux
    integrable partial difference equations \cite{AdlerStartsev1999}.
    In \cite{GarifullinYamilov2012,GarifullinYamilov2015,GSY_DarbouxI}
    they were they where used to find the first integrals of some classes
    of partial difference equations.
    \label{eq:remsym}
    \label{rem:annihilop}
\end{remark}

\subsection{Integrability of difference equations}
\label{sse:int}

Integrability both for continuous and discrete systems can be defined
in different ways, see \cite{whatisintegrability1991,HietarintaBook}
for a complete discussion of the continuous and the discrete case.
After this section, when we will talk about integrability we will
mean \emph{Lioville integrability for variational difference equations}.

To better appreciate the true meaning and power of
Liouville integrability we start from the case of difference
equations of arbitrary order $N$:
\begin{equation}
    x_{n+N} = Q\left( x_{n+N-1},x_{n+N-2},\dots,x_{n} \right),
    \label{eq:ordN}
\end{equation}
where $N$ is not necessarily even and no assumption on variational
structures is made.

Let us assume that there exists $N-1$ functionally independent functions
\begin{equation}
    I_{l} = I_{l} \left( x_{n+N-1},x_{n+N-2},\dots,x_{n} \right),
    \quad
    l=1,\dots,N-1,
    \label{eq:Il}
\end{equation}
called \emph{invariants}, such that
\begin{equation}
    I_{l} \left( x_{n+N},x_{n+N-1},\dots,x_{n+1} \right) 
    =
    I_{l} \left( x_{n+N-1},x_{n+N-2},\dots,x_{n} \right) 
    \label{eq:invdef}
\end{equation}
on the solutions of equation \eqref{eq:ordN}.
Then in principle it is possible
to reduce the difference equation \eqref{eq:ordN} to first order one
by solving the relations:
\begin{equation}
    I_{j} = \kappa_{j},
    \label{eq:intred}
\end{equation}
where $\kappa_{j}$ are the value of the invariants on a set of initial
data.
In such case we say that the difference equation \eqref{eq:ordN}
\emph{integrable}.

\begin{remark}
    We underline that in this subsection both the difference equation 
    \eqref{eq:ordN} and the invariants \eqref{eq:Il} are 
    \emph{assumed to be autonomous}.
    Then, upon translation it is possible to consider any
    \emph{arbitrary choice of $N$ consecutive indices in their definition}.
    We will denote such choice by
    $\Lambda_{N}=\left\{ n_{0}+N-1,\dots,n_{0} \right\}$.
    Observe that $\abs{\Lambda_{N}}=N$.
    \label{rem:indices}
\end{remark}


This definition of integrability is very general.
If some additional structure is present, then the number of invariants
needed for integrability can be significantly lowered.
A special, but relevant case is the one of Poisson difference equations.
Consider the space of functions $f=f( \left\{ x_{n+j} |j\in\Lambda_{N} \right\} )$.
A bilinear operation $\left\{\,,\,\right\}$ on such space satisfying the following
conditions is called a \emph{Poisson bracket} \cite{CapelSahadevan2001, Olver1986}:
\begin{description}
    \item[Skew-symmetry] $\left\{f,g  \right\}=-\left\{ g,f \right\}$.
    \item[Jacobi ideantity] 
        $\left\{ \left\{ f,g \right\},h \right\}+\left\{ \left\{ g,h \right\},f  \right\}+
        \left\{ \left\{ h,f \right\},g \right\}=0$.
\end{description}
A Poisson bracket is completely specified by its action on
the basic functions $x_{n+j}$ for $j\in\Lambda_{N}$, and extended
to generic functions $f$ and $g$ by:
\begin{equation}
    \left\{f,g  \right\}
    =
    \sum_{i,j\in\Lambda_{N}}\pdv{f}{x_{n+i}}\pdv{g}{x_{n+j}} \left\{ x_{n+i},x_{n+j} \right\}.
    \label{eq:deter}
\end{equation}
Two functions $f$ and $g$ such that $\left\{ f,g \right\}=0$ are said to
be \emph{in involution}.
The dimension of the image of a Poisson bracket is called its \emph{rank}.
From skew-symmetry it follows that the rank is an even number, $2r\leq N$.

A difference equation is said to be a Poisson difference equation if
it preserves a Poisson bracket, that is if
\begin{equation}
    \left\{ x_{n+i}',x_{n+j}' \right\} = \left\{ x_{n+i},x_{n+j} \right\}',
    \quad i,j\in\Lambda_{N},
    \label{eq:Poissonpreserve}
\end{equation}
where
\begin{equation}
    x_{n+i}' = x_{n+i+1},\, i\in\Lambda_{N}\setminus\left\{n_{0}+N-1\right\},
    \,
    x_{n+n_{0}+N-1}'= Q\left( \left\{ x_{n+j} |j\in\Lambda_{N} \right\} \right),
    \label{eq:poissonmap}
\end{equation}
and $f\left( \left\{ x_{n+j} |j\in\Lambda_{N} \right\} \right)'=
f\left( \left\{ x_{n+j}' |j\in\Lambda_{N} \right\} \right)$.

Then we have the following characterisation of integrability for
Poisson difference equations:

\begin{theorem}[Discrete Liouville-Poisson theorem \cite{Veselov1991,Bruschietal1991,Maeda1987}]
    If an order $N$ difference equation \eqref{eq:ordN} 
    \emph{preserves a Poisson bracket} of rank $2r$ and possesses 
    \emph{$N-r$ functionally independent invariants}
    in \emph{involution with respect to this Poisson structure},
    then there exists a set of canonical coordinates terms of
    which the difference equations is linear.
    \label{thm:liouvilleintegrability}
\end{theorem}

A difference equation satisfying the hypotheses of Theorem \ref{thm:liouvilleintegrability}
is said to be \emph{Lioville-Poisson integrable}.
If the difference equation has even order $N=2k$ and the Poisson
bracket has \emph{full rank} $k$ it turns out that we are in a
special in which exactly \emph{half} of the invariants are needed
to be in the hypotheses of Theorem \ref{thm:liouvilleintegrability}. 
In such special case we say that the difference equation is
\emph{Liouville integrable}.

Liouville-Poisson integrability requires a ``good'' Poisson bracket.
In \cite{ByrnesHaggarQuispel1999} it was proved that there exists a 
Poisson bracket for any $N$-dimensional volume-preserving map possessing 
$N-2$ invariants.
The obtained Poisson bracket will not be, in general, of maximal rank and
at least an additional invariant is needed to apply Theorem \ref{thm:liouvilleintegrability}.

In this picture variational difference equations play a special
r\^ole, as for them it is \emph{always possible} to find a full rank 
Poisson bracket.
This is possible through a construction called 
\emph{discrete Ostrogradsky transformation} \cite{TranvanderKampQusipel2016}
and it is the content of the following theorem:

\begin{theorem}[Bruschi \emph{et al.} \cite{Bruschietal1991}]
    Assume we are given a variational difference equation of order $2k$
    \eqref{eq:recgen}
    arising from an \emph{autonomous} normal Lagrangian \eqref{eq:lagrgen}.
    Then the change of coordinates with new variables 
    $\left( \vec{q},\vec{p} \right)=\left( q_{1},\dots,q_{N},p_{1},\dots,p_{N} \right)$
     defined through the formula:
    \begin{subequations}
        \begin{align}
            q_i&=x_{n+i-1}, \quad i=1,\dots,k,
            \label{C4E:OstraTranA} 
            \\
            p_i&=T^{-1}\sum_{j=0}^{k-i}T^{-j}\pdv{L}{x_{n+i+j}}, \quad
            i=1,\dots,k,
            \label{C4E:OstraTranB}
        \end{align}    
        \label{eq:C4EOstraTran}
    \end{subequations}
    is well defined and invertible.
    Moreover, the matrix:
    \begin{equation}
        \mathcal{J} = \frac{\partial \left( x_{n+k-1},\dots,x_{n-k} \right)}{\partial \left( \vec{q},\vec{p} \right)}
        \begin{pmatrix}
            \mathbb{O}_{N} & -\mathbb{I}_{N}
            \\
            \mathbb{I}_{N} & \mathbb{O}_{N}
        \end{pmatrix}
        \frac{\partial \left( x_{n+k-1},\dots,x_{n-k} \right)}{\partial \left( \vec{q},\vec{p} \right)}^{T},
        \label{eq:Poissoncan}
    \end{equation}
    is skew-symmetric of rank $2k$ and defines 
    the following full rank Poisson bracket:
    \begin{equation}
        \left\{ x_{n+i},x_{n+j} \right\} = \mathcal{J}_{k-i,k-j},
        \quad
        i,j\in\left\{ k-1,\dots,-k \right\}.
        \label{eq:poissondefcan}
    \end{equation}
    for the variational difference equation \eqref{eq:recgen}.
    \label{thm:ostro}
\end{theorem}

From Theorem \ref{thm:ostro} it follows that integrability for
variational difference equation is proven in Liouville sense if
we are able to produce $k$ functionally independent invariants
in involution with respect to the Poisson bracket \eqref{eq:poissondefcan}.

\begin{remark}
    We remark that in the case of fourth-order difference equations
    the \emph{na\"ive} definition of integrability and Liouville-Poisson
    integrability actually coincide.
    Indeed in the case of fourth-order difference equations a Poisson bracket 
    can have either rank two or four. If it has rank four then, it is full
    rank and we are in the case of Liouville integrability.
    If the Poisson bracket is degenerate with rank two to prove Liouville-Poisson
    integrability according to Theorem \ref{thm:liouvilleintegrability}
    we need $4-1=3$ invariants.
    However, for fourth-order difference equations we need three invatiants
    also to claim integrability in the \emph{na\"ive} sense.
    For this reason we see that in the case of fourth-order differential-difference
    equation the variational structure is much more helpful in proving
    integrability than in the general case, as it really lowers the number
    of invariants needed.
    \label{rem:4thorderLP}
\end{remark}

\section{Main results}
\label{sec:main}

In this section we state and prove our main result on the structure
of additive variational fourth-order difference equations.
We then discuss a general procedure to test if an additive fourth-order
difference equation and present some examples.

\subsection{General results}
It is easy to prove that any additive second-order
difference equation \eqref{eq:add2} is variational
with the following Lagrangian:
\begin{equation}
    L = x_{n}x_{n+1} - \int^{x_{n}} f\left( \xi \right)\ud \xi.
    \label{eq:Ladd2}
\end{equation}
A natural generalisation of \eqref{eq:add2} could be the following
one:
\begin{equation}
    x_{n+2} + x_{n-2} = f\left( x_{n+1},x_{n},x_{n-1} \right).
    \label{eq:addpur2}
\end{equation}
However confronting equation \eqref{eq:addpur2} with the known
examples of variational fourth-order differential equations from
\cite{GJTV_class}, it is clear that this functional form is too narrow.
Indeed, the variational equations given in \cite{GJTV_class} are of 
the following form:
\begin{equation}
    g\left( x_{n+1} \right)x_{n+2} + g\left( x_{n-1} \right) x_{n-2}
    = f\left(x_{n+1},x_{n},x_{n-1}\right).
    \label{eq:rec4addgenGG}
\end{equation}
This gives us the motivation to consider general additive fourth-order 
difference equations as given in \eqref{eq:rec4add}.
We state and prove the following theorem:
\begin{theorem} 
    An additive fourth-order difference 
    equation \eqref{eq:rec4add}
    is Lagrangian if and only it has the following form:
    \begin{equation}
        \begin{aligned}
            g\left( x_{n+1} \right)x_{n+2} &+ \lambda^{2}g\left( x_{n-1} \right) x_{n-2}+
            \lambda g'\left( x_{n} \right) x_{n+1}x_{n-1}
            \\
            &+\pdv{V}{x_{n}}\left( x_{n+1},x_{n} \right)+\lambda\pdv{V}{x_{n}}\left( x_{n},x_{n-1} \right)=0,
        \end{aligned}
        \label{eq:rec4addlagr}
    \end{equation}
    that is:
    \begin{subequations}
        \begin{align}
            f\left( x_{n+1},x_{n},x_{n-1} \right)&=-\lambda^{2}\frac{g\left( x_{n-1} \right)}{g\left( x_{n+1} \right)}
            \label{eq:fform}
            \\
            g\left( x_{n+1} \right)h\left( x_{n+1},x_{n},x_{n-1} \right)&
            \begin{aligned}[t]
                &=
        -\lambda g'\left( x_{n} \right) x_{n+1}x_{n-1}
        \\
        &-\pdv{V}{x_{n}}\left( x_{n+1},x_{n} \right)-\lambda\pdv{V}{x_{n}}\left( x_{n},x_{n-1} \right)=0.
            \end{aligned}
            \label{eq:hform}
        \end{align}
        \label{eq:fhform}
    \end{subequations}
    In that cases the Lagrangian, up to total difference and multiplication
    by a constant is given by:
    \begin{equation}
        L_{n} =\lambda^{-n}\left[ g\left( x_{n+1} \right) x_{n}x_{n+2} + V\left( x_{n+1},x_{n} \right)\right].
        \label{eq:rec4addL}
    \end{equation}
    \label{thm:structure}
\end{theorem}


\begin{proof}
    The only if part is trivial, as using formula \eqref{eq:elgen} it is possible
    to show that the Euler--Lagrange equation corresponding to \eqref{eq:rec4addL}
    is given by equation \eqref{eq:rec4addlagr}.
    Therefore we will concentrate on the proof of the if part.

    To prove the if part we use Theorem \ref{thm:elexist} on an
    additive fourth-order difference equation \eqref{eq:rec4add}
    and its inverse \eqref{eq:rec4addinv}.
    According to Theorem \ref{thm:elexist} a Lagrangian
    $L_{n}$ must satisfy equations \eqref{eq:annihil} and
    \eqref{eq:annihil2}. In the case additive fourth-order difference 
    equation \eqref{eq:rec4add} we have:
    \begin{subequations}
        \begin{align}
            F\left( x_{n+1},x_{n},x_{n-1},x_{n-2} \right) &=
            f\left( x_{n+1},x_{n},x_{n-1} \right) x_{n-2}+h\left( x_{n+1},x_{n},x_{n-1} \right),
            \label{eq:Flin}
            \\
            \widetilde{F}\left( x_{n+2},x_{n+1},x_{n},x_{n-1} \right) &=
            \frac{ x_{n+2}-h\left( x_{n+1},x_{n},x_{n-1} \right)}{%
            f\left( x_{n+1},x_{n},x_{n-1} \right)},
            \label{eq:Ftlin}
        \end{align}
        \label{eq:FFtlin}
    \end{subequations}
    so that, writing explicitly equations \eqref{eq:annihil} and 
    \eqref{eq:annihil2}, we have:
    \begin{subequations}
        \begin{gather}
            \label{eq:cc0bis}
            \begin{gathered}
            {\frac {\partial f}{\partial x_{{n-1}}}} 
            \left( x_{{n+1}},x_{{n}},x_{{n-1}} \right)
            {\frac {\partial ^{2}L_{{n-2}}}{\partial x_{{n-2}}\partial x_{{n}}}} 
            \left( x_{{n}},x_{{n-1}},x_{{n-2}} \right) 
            \\
                -f \left( x_{{n+1}},x_{{n}},x_{{n-1}} \right) 
                {\frac {\partial ^{3}L_{{n-2}}}{\partial x_{{n-1}}\partial x_{{n-2}}\partial x_{{n}}}} 
                \left( x_{{n}},x_{{n-1}},x_{{n-2}} \right) 
            \\
            + \left[  
                {\frac {\partial f}{\partial x_{{n-1}}}} 
                \left( x_{{n+1}},x_{{n}},x_{{n-1}} \right) x_{{n-2}}
                +{\frac {\partial h}{\partial x_{{n-1}}}} 
            \left( x_{{n+1}},x_{{n}},x_{{n-1}} \right)  \right]
                \times
                \\
                {\frac {\partial ^{3} L_{{n-2}}}{\partial x_{{n-2}}^{2}\partial x_{{n}}}} 
                \left( x_{{n}},x_{{n-1}},x_{{n-2}}\right)=0,
                \\
            \end{gathered}
            \\
            \begin{gathered}
                \left[  
                    \left( x_{{n+2}}-f \left( x_{{n+1}},x_{{n}},x_{{n-1}}\right)  \right) 
                    {\frac {\partial f}{\partial x_{{n+1}}}}
                \left( x_{{n+1}},x_{{n}},x_{{n-1}} \right) \right.
                    \\
                    \left.
                    +f \left( x_{{n+1}},x_{{n}},x_{{n-1}}\right) 
                    {\frac {\partial h}{\partial x_{{n+1}}}} 
                \left( x_{{n+1}},x_{{n}},x_{{n-1}} \right)  \right]
                {\frac {\partial ^{3}L_{{n}}}{\partial x_{{n+2}}^{2}\partial x_{{n}}}}  
                \left( x_{{n+2}},x_{{n+1}},x_{{n}}\right)
                \\
                 + \left( {\frac {\partial ^{3}L_{{n}}}{\partial x_{{n+2}}\partial x_{{n+1}}\partial x_{{n}}}} \left( x_{{n+2}},x_{{n+1}},x_{{n}} \right)  \right) 
                 f \left( x_{{n+1}},x_{{n}},x_{{n-1}} \right) 
                 \\
                 + {\frac {\partial f}{\partial x_{{n+1}}}} 
                \left( x_{{n+1}},x_{{n}},x_{{n-1}} \right) 
                {\frac {\partial ^{2} L_{n}}{\partial x_{{n+2}}\partial x_{{n}}}}
                \left( x_{{n+2}},x_{{n+1}},x_{{n}} \right)=0.
            \end{gathered}
            \label{eq:cc0invbis}
        \end{gather}
        \label{eq:cc0bistot}
    \end{subequations}
    
    Let us start from equation \eqref{eq:cc0bis}.
    Since the functions $f$ and $h$ are unknown we cannot
    use the standard solving technique shown in \cite{Gubbiotti_dcov}.
    On the other hand we need to use the fact $f$ and $h$ depend
    on $x_{n+1},x_{n},x_{n-1}$ while $L_{n-2}=L_{n-2}\left( x_{n},x_{n-1},x_{n-2} \right)$.
    We can eliminate $L_{n-2}$ solving with respect to its derivatives
    and then differentiating with respect to $x_{n+1}$.
    To completely eliminate it we need to repeat this process three
    times.
    This yields the following equation (since $f$ and $h$ depend on the
    same variables we drop the explicit dependence on $x_{n+1}$, $x_{n}$ and $x_{n-1}$):
    \begin{equation}
        \begin{gathered}
        {\frac {\partial ^{3} f}{\partial x_{{n+1}}^{2}\partial x_{{n-1}}}}
        {\frac {\partial f}{\partial x_{{n+1}}}}
        {\frac {\partial h}{\partial x_{{n-1}}}}
        - {\frac {\partial ^{3} f}{\partial x_{{n+1}}^{2}\partial x_{{n-1}}}}
            {\frac {\partial ^{2} h}{\partial x_{{n+1}}\partial x_{{n-1}}}}g 
            \\
            + {\frac {\partial ^{3} h}{\partial {x_{{n+1}}}^{2}\partial x_{{n-1}}}} 
            {\frac {\partial ^{2}f}{\partial x_{{n+1}}\partial x_{{n-1}}}} g 
            - {\frac {\partial ^{3} h}{\partial x_{{n+1}}^{2}\partial x_{{n-1}}}}
            {\frac {\partial f}{\partial x_{{n+1}}}}
            {\frac {\partial f}{\partial x_{{n-1}}}} 
            \\
            - {\frac{\partial ^{2} f}{\partial x_{{n+1}}^{2}}}
                {\frac {\partial ^{2} f}{\partial x_{{n+1}}\partial x_{{n-1}}}}
                    {\frac {\partial h}{\partial x_{{n-1}}}}
            +{\frac {\partial ^{2} f}{\partial x_{{n+1}}^{2}}}
            {\frac {\partial ^{2} h}{\partial x_{{n+1}}\partial x_{{n-1}}}}
            {\frac {\partial h}{\partial x_{{n-1}}}}=0.
        \end{gathered}
        \label{eq:cc0fin}
    \end{equation}
    Using the CAS \texttt{Maple 2016}\ to solve equation \eqref{eq:cc0fin}
    we find that the solution is actually independent of $h$ and has
    the following form:
    \begin{equation}
        f\left( x_{n+1},x_{n},x_{n-1} \right) =
        G_{+}\left( x_{n+1} \right) G\left( x_{n} \right) G_{-}\left( x_{n-1} \right).
        \label{eq:gsol1}
    \end{equation}
    Going back to equation \eqref{eq:cc0bis}, if we solve with respect
    to $\partial^{3} L_{n-2} /\partial x_{n-2}^{2}\partial x_{n}$ and differentiating
    with respect to $x_{n+1}$ we find the following simple compatibility condition:
    \begin{equation}
        \begin{gathered}
        G\left( x_{n} \right)
        \left[ G_{+}'(x_{n+1})\pdv{h}{x_{n-1}} (x_{n+1}, x_{n}, x_{n-1})
        -G_{+}(x_{n+1})\pdv{h}{x_{n-1},x_{n+1}} (x_{n+1}, x_{n}, x_{n-1})  \right]
        \times
        \\
        \left[ G_{-}(x_{n-1})\pdv{l_{n-2}}{x_{n-1},x_{n}} (x_{n}, x_{n-1}, x_{n-2})
            -\pdv{l_{n-2}}{x_{n}} (x_{n}, x_{n-1}, x_{n-2}) G_{-}'(x_{n-1}) \right]=0,
        \end{gathered}
        \label{eq:cc0trisbis}
    \end{equation}
    where we defined:
    \begin{equation}
        l_{n-2}\left( x_{n},x_{n-1},x_{n-2} \right)\equiv 
        \pdv{L_{n-2}}{x_{n-2}}\left( x_{n},x_{n-1},x_{n-2} \right).
        \label{eq:lmmdef}
    \end{equation}

    Equation \eqref{eq:cc0trisbis} has three factors which can be annihilated
    separately.
    The first factor gives $G\left( x_{n} \right)=0$, that is $f\equiv0$,
    which is not allowed.
    Therefore from \eqref{eq:cc0trisbis} we can choose to fix either
    $f$ or $l_{n-2}$. We will now address these two possibilities.

    \subsection*{Fix $l_{n-2}$ from \eqref{eq:cc0trisbis}}
    Solving the second factor in \eqref{eq:cc0trisbis} we obtain
    the following value for $l_{n-2}$:
    \begin{equation}
        l_{n-2}\left( x_{n},x_{n-1},x_{n-2} \right) 
        = l_{1,n-2}\left( x_{n},x_{n-2} \right)G_{-}\left( x_{n-1} \right)
        +l_{2,n-2}\left( x_{n-1},x_{n-2} \right).
        \label{eq:ls2}
    \end{equation}
    Inserting \eqref{eq:ls2} into \eqref{eq:cc0bis} we obtain
    the following equation:
    \begin{equation}
        \begin{aligned}
        \left(
        G_{+} \left( x_{{n+1}} \right)  
        G_{-}'\left( x_{{n-1}} \right) 
        G\left( x_{{n}} \right) x_{{n-2}}
        +{\frac {\partial h}{\partial x_{{n-1}}}}\left( x_{{n+1}},x_{{n}},x_{{n-1}} \right)  \right) 
        &\times
        \\
        {\frac {\partial ^{2} l_{1,n-2}}{\partial x_{{n-2}}\partial x_{{n}}}}  
        \left( x_{{n}},x_{{n-2}} \right)&=0. 
        \end{aligned}
        \label{eq:fann}
    \end{equation}
    Again we have two factors we can choose to annihilate.
    The first factors, since no function depends on $x_{n-2}$ is equivalent
    to the following equations:
    \begin{equation} 
        G_{+} \left( x_{{n+1}} \right)  
        G_{-}'\left( x_{{n-1}} \right) 
        G\left( x_{{n}} \right) = 0,
        \quad 
        {\frac {\partial h}{\partial x_{{n-1}}}}\left( x_{{n+1}},x_{{n}},x_{{n-1}} \right)=0.
        \label{eq:fanncoff}
    \end{equation}
    The first equation imply $G_{-}\left( x_{n-1} \right)=\text{constat}$ and
    the second one imply $h=h\left( x_{n+1},x_{n} \right)$.
    This is not allowed as the equation will be independent of $x_{n-1}$.
    Therefore we are forced to annihilate the second factor.
    This implies:
    \begin{equation}
       l_{1,n-2}\left( x_{{n}},x_{{n-2}} \right) 
       = l_{1,1,n-2}\left( x_{n-2} \right)+l_{1,2,n-2}\left( x_{n} \right).
        \label{eq:l1s2}
    \end{equation}
    Inserting this into \eqref{eq:ls2} and using the arbitrariness of
    $l_{2,n-2}$ we can write:
    \begin{equation}
        l_{n-2}\left( x_{n},x_{n-1},x_{n-2} \right) 
        = l_{1,2,n-2}\left( x_{n} \right)G_{-}\left( x_{n-1} \right)
        +\pdv{l_{2,n-2}}{x_{n-2}}\left( x_{n-1},x_{n-2} \right).
        \label{eq:ls2bis}
    \end{equation}

    Using the definition of $l_{n-2}$ \eqref{eq:lmmdef} and the
    fact that discrete Lagrangians are defined only up to total
    difference, from formula \eqref{eq:ls2bis} we obtain the following
    form of the Lagrangian:
    \begin{equation}
        L_{n}\left( x_{n+2},x_{n+1},x_{n} \right) 
        = l_{1,2,n-2}\left( x_{n+2} \right)G_{-}\left( x_{n+1} \right) x_{n}
        +l_{2,n}\left( x_{n+1},x_{n} \right).
        \label{eq:Ls2}
    \end{equation}
    The Euler--Lagrange equation corresponding to \eqref{eq:Ls2}, upon
    substitution of equation \eqref{eq:rec4addinv} are:
    \begin{equation}
        \begin{gathered}
        {\frac {\partial l_{2,n}}{\partial x_{{n}}}} \left( x_{{n}},x_{{n+1}} \right) 
        +G_{-} \left( x_{{n+1}} \right) l_{1,2,n} \left( x_{{n+2}} \right) 
        +{\frac {\partial l_{2,n-1}}{\partial x_{{n}}}}\left( x_{{n-1}},x_{{n}} \right) 
        \\
        +G_{-}'\left( x_{{n}} \right) l_{1,2,n-1}\left( x_{{n+1}} \right) x_{{n-1}}
        =l_{1,2,n-2}' \left( x_{{n}} \right)
        \frac { h \left( x_{{n+1}},x_{{n}},x_{{n-1}}\right)-x_{{n+2}}}{ G_{+}\left( x_{{n+1}} \right) G \left( x_{{n}} \right)}.
        \end{gathered}
        \label{eq:cc2}
    \end{equation}
    Differentiating equation \eqref{eq:cc2} with respect to $x_{n+2}$ twice
    we obtain:
    \begin{equation}
        G_{-} \left( x_{{n+1}} \right) l_{1,2,n}'' \left( x_{{n+2}} \right)=0.
        \label{eq:cc2dd}
    \end{equation}
    Using the usual argument, we obtain that we need to annihilate the second
    factor, which gives:
    \begin{equation}
        l_{1,2,n} \left( x_{{n+2}} \right) = C_{1,n} x_{n+2}+C_{2,n},
        \label{eq:l12s}
    \end{equation}
    where $C_{1,n}$ and $C_{2,n}$ are two functions depending on
    $n$ alone.
    Substituting back in equation \eqref{eq:cc2} and applying the differential
    operator
    \begin{equation}
        \dv{}{x_{n}}\left( \frac{1}{G\left( x_{n} \right)} \dv{}{x_{n+2}} \right)
        \label{eq:diffop}
    \end{equation}
    we obtain:
    \begin{equation}
        C_{1,n-2} \frac{G'\left( x_{n} \right)}{G^{2}\left( x_{n} \right)} =0.
        \label{eq:cc2dd2}
    \end{equation}
    Since $C_{1,n-2}\neq0$ we obtain $G\left( x_{n} \right) =1/K_{1}$ where $K_{1}$
    is a constant.
    Inserting this value into \eqref{eq:cc2} we obtain:
    \begin{equation}
        \begin{gathered}
        {\frac {\partial l_{2,n}}{\partial x_{{n}}}} \left( x_{{n}},x_{{n+1}} \right) 
        +G_{-} \left( x_{{n+1}} \right) \left( C_{1,n} x_{{n+2}}+C_{2,n} \right) 
        +{\frac {\partial l_{2,n-1}}{\partial x_{{n}}}}\left( x_{{n-1}},x_{{n}} \right) 
        \\
        + G_{-}'\left( x_{{n}} \right) \left(C_{1,n-1}x_{{n+1}}+C_{2,n-1}\right) x_{{n-1}}
        =
        C_{1,n-2} K_{1} 
        \frac { h \left( x_{{n+1}},x_{{n}},x_{{n-1}}\right)-x_{{n+2}}}{ G_{+}\left( x_{{n+1}} \right)}.
        \end{gathered}
        \label{eq:cc2tris}
    \end{equation}
    We can take the coefficient with respect to $x_{n+2}$ and we obtain:
    \begin{equation}
        G_{-} \left( x_{{n+1}} \right)  C_{1,n} =- \frac {C_{1,n-2} K_{1}}{ G_{+}\left( x_{{n+1}} \right)}.
        \label{eq:cc2trisxpp}
    \end{equation}
    We can rewrite this equation as:
    \begin{equation}
        G_{-} \left( x_{{n+1}} \right) G_{+}\left( x_{{n+1}} \right)  =- \frac {C_{1,n-2} K_{1}}{C_{1,n}}.
        \label{eq:cc2trisxppbis}
    \end{equation}
    Since $K_{1}$ is a constant, upon differentiation with respect to $x_{n+1}$,
    there exists a constant $q\in\R\setminus{\left\{ 0 \right\}}$ such that:
    \begin{equation}
        G_{+}\left( x_{{n+1}} \right)  = \frac{K_{1}}{qG_{-} \left( x_{{n+1}} \right)},
        \quad\text{and}\quad
         C_{1,n}= q C_{1,n-2}.
        \label{eq:GpC1neq}
    \end{equation}
    Using conditions \eqref{eq:GpC1neq} into 
    \eqref{eq:cc2tris} e obtain:
    \begin{equation}
        \begin{gathered}
        {\frac {\partial l_{2,n}}{\partial x_{{n}}}} \left( x_{{n}},x_{{n+1}} \right) 
        +G_{-} \left( x_{{n+1}} \right) C_{2,n}
        +{\frac {\partial l_{2,n-1}}{\partial x_{{n}}}}\left( x_{{n-1}},x_{{n}} \right) 
        \\
        + G_{-}'\left( x_{{n}} \right) \left(C_{1,n-1}x_{{n+1}}+C_{2,n-1}\right) x_{{n-1}}
        =
        C_{1,n-2} G_{-}\left( x_{{n+1}} \right)
        h \left( x_{{n+1}},x_{{n}},x_{{n-1}}\right) .
        \end{gathered}
        \label{eq:cc2quadrbis}
    \end{equation}
    Differentiating with respect to $x_{n+1}$ and $x_{n-1}$ we obtain a
    PDE for $h$ which can be solved to give:
    \begin{equation}
        \begin{aligned}
        h \left( x_{{n+1}},x_{{n}},x_{{n-1}}\right)
        &=
        \frac{h_{1}\left( x_{n},x_{n-1} \right) +h_{2}\left( x_{n+1},x_{n} \right)}{%
            G_{-}\left( x_{n+1} \right)}
            \\
            &-\frac{C_{1,n-1}}{qC_{1,n-2}} 
            \frac{G_{-}\left( x_{n} \right)x_{n+1}x_{n-1}}{G_{-}\left( x_{n+1} \right)}.
        \end{aligned}
        \label{eq:fgen}
    \end{equation}
    Since $f$ must not depend explicitly on $n$, we must impose that the
    coefficient $F_{n} =C_{1,n-1}/C_{1,n-2}$ is $n$ independent, that is
    it is a total difference. Using again equation \eqref{eq:GpC1neq} we obtain:
    \begin{equation}
        C_{1,n-1}^{2} - q C_{1,n-2}^{2} = 0.
        \label{eq:fincond}
    \end{equation}
    This implies $q>0$, that is $q=\lambda^{2}$ for some $\lambda\in\R\setminus{ \left\{ 0 \right\} }$, 
    and then
    \begin{equation}
        C_{1,n}^{\pm} = A \left(\pm\lambda\right)^{n}
        \label{eq:C1sol}
    \end{equation}
    with $A\in\R$ a constant.
    However, due to the arbitrariness of $\lambda$ we can consider only the
    solution $C_{1,n}^{+}$.
    Indeed, $\lambda$ can be negative and the cases with $C_{1,n}^{-}$ just follow
    from the substitution $\lambda\to-\lambda$.
    Therefore we drop the superscript $+$ in \eqref{eq:C1sol}.
    This reasoning implies that the $f$ in \eqref{eq:fgen} assumes the following
    form:
    \begin{equation}
        h \left( x_{{n+1}},x_{{n}},x_{{n-1}}\right)
        =
        \frac{h_{1}\left( x_{n},x_{n-1} \right) + h_{2}\left( x_{n+1},x_{n} \right)}{%
            G_{-}\left( x_{n+1} \right)}
            -\frac{G_{-}\left( x_{n} \right)x_{n+1}x_{n-1}}{\lambda G_{-}\left( x_{n+1} \right)}.
        \label{eq:fs1}
    \end{equation}

    We can finally insert \eqref{eq:fs1} into \eqref{eq:cc2quadrbis}
    and obtain:
    \begin{equation}
        \begin{gathered}
        {\frac {\partial l_{2,n}}{\partial x_{{n}}}} \left( x_{{n}},x_{{n+1}} \right) 
        +G_{-} \left( x_{{n+1}} \right) C_{2,n}
        +{\frac {\partial l_{2,n-1}}{\partial x_{{n}}}}\left( x_{{n-1}},x_{{n}} \right) 
        \\
        + G_{-}'\left( x_{{n}} \right) C_{2,n-1} x_{{n-1}}
        =
        -A \lambda^{n}
        \left[
        f_{1}\left( x_{n},x_{n-1} \right) + f_{2}\left( x_{n+1},x_{n} \right) \right].
        \end{gathered}
        \label{eq:cc3}
    \end{equation}
    Differentiating with respect to $x_{n+1}$ we obtain a linear PDE
    for $l_{2,n}\left( x_{n+1},x_{n} \right)$.
    Solving such equation we obtain the following form for this function:
    \begin{equation}
        \begin{aligned}
            l_{2,n}\left( x_{n+1},x_{n} \right)
            &=
        l_{2,2,n}\left( x_{n+1} \right) + l_{2,1}\left( x_{n} \right)
        \\
        &- A\lambda^{n} \int^{x_{n}} f_{2}\left( x_{n+1}, \xi \right)\ud \xi
        -C_{2,n} x_{n} G_{-}\left( x_{n+1} \right).
        \end{aligned}
        \label{eq:l2ns1}
    \end{equation}
    From the form of the Lagrangian function, using the property of
    equivalence, we can remove the arbitrary function $l_{2,2,n}\left( x_{n+1} \right)$
    and keep only $l_{2,1,n}\left( x_{n} \right)$.
    So, in \eqref{eq:cc3} this yields:
    \begin{equation}
        \frac{l_{2,1,n}'\left( x_{n} \right)}{A \lambda^{n}}
        +f_{1}\left( x_{n},x_{n-1} \right)
        =\frac{1}{\lambda}\int^{x_{n-1}} \pdv{f_{2}}{x_{n}}\left( x_{n},\xi \right)\ud\xi.
        \label{eq:ccf1}
    \end{equation}
    Differentiation with respect to $x_{n-1}$ yields the following
    equation:
    \begin{equation}
        \pdv{f_{1}}{x_{n-1}}\left( x_{n},x_{n-1} \right)
        =\frac{1}{\lambda}\pdv{f_{2}}{x_{n}}\left( x_{n},x_{n-1} \right).
        \label{eq:closrelpot}
    \end{equation}
    Equation \eqref{eq:closrelpot} stimulates the introduction
    of a \emph{potential} function $V=V\left( x_{n},x_{n-1} \right)$ such
    that:
    \begin{equation}
        f_{1}\left( x_{n},x_{n-1} \right) =
        \frac{1}{\lambda} \pdv{V}{x_{n}}\left( x_{n},x_{n-1} \right),
        \quad
        f_{2}\left( x_{n},x_{n-1} \right) = 
        \pdv{V}{x_{n-1}}\left( x_{n},x_{n-1} \right).
        \label{eq:F1F2pot}
    \end{equation}
    Using such potential we have that equation
    \eqref{eq:closrelpot} is identically satisfied, while
    \eqref{eq:ccf1} reduces
    to $l_{2,1,n}'\left( x_{n} \right)=0$.
    This implies $l_{2,1,n}\left( x_{n} \right)= C_{3,n}$, but this function
    of $n$ can be removed from the Lagrangian as it is a total difference.

    Summing up, we obtained that if an additive fourth-order difference equation 
    \eqref{eq:rec4add} is Lagrangian then it has the following form:
    \begin{equation}
        \begin{aligned}
            G_{-}\left( x_{n+1} \right)x_{n+2} &+ \lambda^{2} G_{-}\left( x_{n-1} \right) x_{n-2}+
            \lambda G_{-}'\left( x_{n} \right) x_{n+1}x_{n-1}
            \\
            &+\pdv{V}{x_{n}}\left( x_{n+1},x_{n} \right)+\lambda\pdv{V}{x_{n}}\left( x_{n},x_{n-1} \right)=0.
        \end{aligned}
        \label{eq:rec4lagrfin}
    \end{equation}
    Letting $g\equiv G_{-}$ equation \eqref{eq:rec4addlagr} follows.
    The constant $K_{1}$ appearing in the Lagrangian can be scaled away
    and we obtain the Lagrangian \eqref{eq:rec4addL}.

    \subsection*{Fix $h$ from \eqref{eq:cc0trisbis}}
    If we fix $h$ from \eqref{eq:cc0trisbis} we obtain:
    \begin{equation}
        h\left( x_{n+1},x_{n},x_{n-1} \right)=
        h_{+}\left( x_{n+1},x_{n} \right)
        +G_{+}\left( x_{n+1} \right) h_{-}\left( x_{n},x_{n-1} \right).
        \label{eq:Ffixed}
    \end{equation}
    After a long calculation which follows the same strategy outlined
    in the case when we fix $l_{n-2}$ from \eqref{eq:cc0trisbis} we find
    that this case implies $G\equiv 0$, and so it is impossible.

    With this we are done with the proof.
\end{proof}

Theorem \ref{thm:structure} characterises completely the additive Lagrangian
fourth-order difference equations \eqref{eq:rec4add}.
An immediate corollary of this theorem is the following one:

\begin{corollary} 
    An additive fourth-order difference 
    equation \eqref{eq:rec4add}
    admits an \emph{autonomous} Lagrangian if and only it has 
    the following form:
    \begin{equation}
        \begin{aligned}
            g\left( x_{n+1} \right)x_{n+2} &+ g\left( x_{n-1} \right) x_{n-2}+
            g'\left( x_{n} \right) x_{n+1}x_{n-1}
            \\
            &+\pdv{V}{x_{n}}\left( x_{n+1},x_{n} \right)+\pdv{V}{x_{n}}\left( x_{n},x_{n-1} \right)=0.
        \end{aligned}
        \label{eq:rec4addlagraut}
    \end{equation}
    In that cases the Lagrangian, up to total difference and multiplication
    by a constant is given by:
    \begin{equation}
        L = g\left( x_{n+1} \right) x_{n}x_{n+2} + V\left( x_{n+1},x_{n} \right).
        \label{eq:rec4addLaut}
    \end{equation}
    \label{cor:structureaut}
\end{corollary}

\begin{proof}
    Trivially follows from Theorem \ref{thm:structure} substituting 
    $\lambda=1$ in formul\ae\ \eqref{eq:rec4addlagr} and \eqref{eq:rec4addL}.
\end{proof}

We choose to present corollary \ref{cor:structureaut} as a separate
result, as in section \ref{sec:int} we will discuss the integrability
properties of a subclass with autonomous Lagrangian.

Theorem \ref{thm:structure} gives also a practical test to establish
whether or not a given additive fourth-order equation \eqref{eq:rec4add} 
is Lagrangian without having to apply the full algorithm of \cite{Gubbiotti_dcov}.
That is, given an additive fourth-order difference equation \eqref{eq:rec4add}
the test runs as follows:
\begin{enumerate}
    \item Write the equation clearing the denominators:
        \begin{equation}
            A\left( x_{n+1},x_{n},x_{n-1} \right) x_{n+2}+
            B\left( x_{n+1},x_{n},x_{n-1} \right) x_{n-2}
            +C\left( x_{n+1},x_{n},x_{n-1} \right)=0.
            \label{eq:dencleared}
        \end{equation}
    \item In order to be in the form \eqref{eq:rec4addlagr}
        the functions $A$ and $B$ needs to be of the following form:
        \begin{subequations}
            \begin{align}
            A\left(x_{n+1},x_{n},x_{n-1}  \right) &= K\left( x_{n} \right) g\left( x_{n+1} \right),
            \\
            B\left(x_{n+1},x_{n},x_{n-1}  \right) &= \lambda^{2} K\left( x_{n} \right) g\left( x_{n-1} \right),
            \end{align}
            \label{eq:ABfact}
        \end{subequations}
        for some function $K=K\left( \eta \right)$ and $g=g\left( \xi \right)$ and constant
        $\lambda$.
    \item Using equation \eqref{eq:ABfact} we divide equation 
        \eqref{eq:dencleared} by $K=K\left( x_{n} \right)$ and using the
        definition of $g$ we rewrite equation \eqref{eq:dencleared}
        as:
        \begin{equation}
            \begin{aligned}
            g\left( x_{n+1}\right) x_{n+2} &+ \lambda^{2} g\left( x_{n-1} \right) x_{n-2}
            +\lambda g'\left( x_{n} \right)x_{n+1}x_{n-1}
            \\
            &+R\left( x_{n+1},x_{n},x_{n-1} \right)=0.
            \end{aligned}
            \label{eq:dencleared2}
        \end{equation}
        with a new function $R=R\left( \xi,\eta,\zeta \right)$\footnote{In this section
            and in the next ones we will indicate various placeholder variables
            with Greek letters $\xi$, $\eta$, $\zeta$\dots. We will use these placeholders
            variables when making statements on functions which might have different
        arguments, e.g. the function $g=g\left( \xi \right)$ in equation \eqref{eq:rec4addlagr}.}.
    \item To be in the form \eqref{eq:rec4addlagr} we need to check that:
        \begin{equation}
            \pdv{R}{x_{n+1},x_{n-1}}\left( x_{n+1},x_{n},x_{n-1} \right) \equiv 0.
            \label{eq:Rcond}
        \end{equation}
    \item If the function $R$ in equation \eqref{eq:dencleared2} satisfies
        condition \eqref{eq:Rcond} it implies that we can write:
        \begin{equation}
            \begin{aligned}
                g\left( x_{n+1}\right) x_{n+2}&+\lambda^{2} g\left( x_{n-1} \right) x_{n-2}
            +\lambda g'\left( x_{n} \right)x_{n+1}x_{n-1}
            \\
            &+M\left( x_{n+1},x_{n} \right)+N\left( x_{n},x_{n-1} \right)=0.
            \end{aligned}
            \label{eq:dencleared3}
        \end{equation}
    \item Comparing again with equation \eqref{eq:rec4addlagr} we have that
        the two functions $M=M\left( \xi,\eta \right)$ and 
        $N=N\left(\xi,\eta  \right)$ need to satisfy the following closure
        relation:
        \begin{equation}
            \lambda
            \pdv{M}{\xi}\left( \xi,\eta \right) =
            \pdv{N}{\eta}\left( \xi,\eta \right).
            \label{eq:closureMN}
        \end{equation}
    \item If the closure relation \eqref{eq:closureMN} is satisfied
        then equation \eqref{eq:dencleared} is in the form \eqref{eq:rec4addlagr},
        therefore it is Lagrangian.
        The function $V$ can be computed using from the following integral:
        \begin{equation}
            V\left( x_{n+1},x_{n} \right) = 
            \int_{\Gamma} M\left( x_{n+1},x_{n} \right)\ud x_{n} +
            \lambda^{-1} N\left( x_{n+1},x_{n} \right) \ud x_{n+1},
            \label{eq:Vcomp}
        \end{equation}
        on a properly chosen path $\Gamma\subset \R^{2}$.
\end{enumerate}

\begin{remark}
    In the above discussion we tacitly assumed
    that the functions $M=M\left( \xi,\eta \right)$ and 
    $N=N\left(\xi,\eta  \right)$ were defined on some simply-connected
    domain $D\subset \R^{2}$, e.g. a star-shaped domain.
    In practice we need to check to check this assumption in order
    to carry out the last step of this test.
    If this hypothesis is not satisfied we cannot use formula
    \eqref{eq:Vcomp}, but we need to directly solve the
    overdetermined system of partial differential equations:
    \begin{equation}
        \pdv{V}{x_{n}}=M(x_{n+1},x_{n}), \quad
        \lambda\pdv{V}{x_{n+1}}=N(x_{n+1},x_{n}).
        \label{eq:Vdiff}
    \end{equation}
    A simple example of this occurrence is given by the following
    additive fourth-order difference equation:
    \begin{equation}
       x_{n+1} x_{n+2}+x_{n-1} x_{n+1}+ x_{n-1} x_{n-2}
       -\frac{x_{n+1}}{x_{n}^2+x_{n+1}^2}+\frac{x_{n-1}}{x_{n}^2+x_{n-1}^2}=0.
        \label{eq:ex0}
    \end{equation}
    \label{re:assmp}
    In this case it easy to see, as the denominator are already cleared,
    that $g\left( \xi \right) =\xi$, $\lambda^{2}=1$ and:
    \begin{equation}
        R\left( x_{n+1},x_{n},x_{n-1} \right)=
        -\frac{x_{n+1}}{x_{n}^2+x_{n+1}^2}+\frac{x_{n-1}}{x_{n}^2+x_{n-1}^2}.
        \label{eq:Rex0}
    \end{equation}
    The condition \eqref{eq:Rcond} is satisfied, and we are left with:
    \begin{equation}
        M\left( \xi,\eta \right)=-\frac{\xi}{\xi^{2}+\eta^2},
        \quad
        N\left( \xi,\eta \right)=\frac{\eta}{\xi^{2}+\eta^2}.
        \label{eq:MNex0}
    \end{equation}
    The closure condition gives $\lambda=1$.
    However, since the functions $M$ and $N$ are defined in 
    the multiply-connected
    domain $D=\R^{2}\setminus \left\{ 0 \right\}$, it is not enough.
    Indeed, it is known that it is not possible to construct the function $V$ 
    using formula \eqref{eq:Vcomp} as the line intengral depends 
    on the path \cite{Dubrovin1990}.
    However, the function
    \begin{equation}
        V\left( x_{n+1},x_{n} \right) = \arctan\left( \frac{x_{n+1}}{x_{n}} \right),
        \label{eq:Vex0}
    \end{equation}
    satisfies the partial differential equation \eqref{eq:Vdiff}.
    Therefore the additive fourth-order difference equation
    \eqref{eq:ex0} is variational with the following Lagrangian:
    \begin{equation}
        L= x_{n}x_{n+1}x_{n+2}+ \arctan\left( \frac{x_{n+1}}{x_{n}} \right).
        \label{eq:Lex0}
    \end{equation}
\end{remark}

\subsection{Examples}

We now discuss three explicit examples of the usage of the test we presented.
In particular in example \ref{ex:filter} and \ref{ex:lineq} we show 
how the test derived from Theorem
\ref{thm:structure} can be used to filter out Lagrangian examples
out of parametric families of equations.

\begin{example}
    Consider the following fourth-order difference equation:
    \begin{equation}
        \frac{x_{n+2}}{x_{n-1}^3}+\frac{x_{n-2}}{x_{n+1}^3}
        +\frac{1}{x_{n-1}^2 x_{n+1}^2}\left[
            3 x_{n}^2-\frac{\mu}{(x_{n}^2-1) x_{n-1} x_{n+1}}
    \right]=0. 
        \label{eq:ex1}
    \end{equation}
    Taking the numerator we find:
    \begin{equation}
        A = (x_{n}^2-1) x_{n+1}^{3},\quad
        B = (x_{n}^2-1) x_{n-1}^{3}
        \label{eq:ABex1}
    \end{equation}
    therefore $K=x_{n}^2-1$, $g=\xi^{3}$ and $\lambda^{2}=1$.
    With this definitions we can rewrite equation
    \eqref{eq:ex1} as:
    \begin{equation}
        x_{n+1}^3x_{n+2}+x_{n-1}^3x_{n-2}
        + 3 \lambda x_{n}^{2} x_{n+1} x_{n-1}
        =R\left( x_{n+1},x_{n},x_{n-1} \right),
        \label{eq:ex1b}
    \end{equation}
    with:
    \begin{equation}
        R \left( x_{n+1},x_{n},x_{n-1} \right)= 
        -\frac{\mu}{x_n^2-1}-3 x_{n-1} x_{n+1} (\lambda-1) x_{n}^2.
        \label{eq:Rex1}
    \end{equation}
    The compatibility condition \eqref{eq:Rcond} gives $\lambda=1$,
    and implies:
    \begin{equation}
        M\left( \xi,\eta \right)  = 
        -\frac{1}{2}\frac{\mu}{\xi^2-1},
        \quad
        N\left( \xi,\eta \right) =
        -\frac{1}{2}\frac{\mu}{\eta^2-1}.
        \label{eq:MNex1}
    \end{equation}
    We can think of the functions $M$ and $N$ as defined
    on the star-shaped domain $D = \left( -1,1 \right)\times\left( -1,1 \right)$
    and compute the function $V$ with formula \eqref{eq:Vcomp}:
    \begin{equation}
        V\left( x_{n+1},x_{n} \right) =
        \frac{\mu}{2}\left[ \arctanh\left( x_{n} \right)+\arctanh\left( x_{n+1} \right) \right].
        \label{eq:Vex1}
    \end{equation}
    The Lagrangian for equation \eqref{eq:ex1} is then given by:
    \begin{equation}
        L=
        x_{n+1}^{3} x_{n} x_{n+2}+
        \frac{\mu}{2}\left[ \arctanh\left( x_{n} \right)+\arctanh\left( x_{n+1} \right) \right].
        \label{eq:Lex1}
    \end{equation}
\end{example}

\begin{example}
    \label{ex:filter}
    Consider the family of fourth-order difference equations:
    \begin{equation}
       x_{n-1}^2 x_{n-2}+x_{n+1}^2 x_{n+2}
       +\frac{1}{1-x_{n}}+x_{n} (a_{02} x_{n-1}^2+a_{11} x_{n-1} x_{n+1}+a_{20} x_{n+1}^2)=0, 
        \label{eq:ex2}
    \end{equation}
    depending parametrically on the three parameters $a_{ij}$, $i+j=2$.
    We will find for which values of these parameters equation \eqref{eq:ex2}
    is Lagrangian.

    First of all we notice that equation \eqref{eq:ex2} has already the numerators
    cleared and that $A=x_{n+1}^{2}$, $B=x_{n-1}^{2}$.
    It follows that $K=1$, $g=\xi^{2}$ and $\lambda^{2}=1$. 
    We can then write down equation \eqref{eq:ex2} as:
    \begin{equation}
        x_{n-1}^2 x_{n-2}+x_{n+1}^2 x_{n+2}+2\lambda x_{n}x_{n+1}x_{n-1}+
        R\left( x_{n+1},x_{n},x_{n-1} \right)=0,
        \label{eq:ex2b}
    \end{equation}
    where the function $R$ is given by:
    \begin{equation}
       R = \frac{1}{1-x_{n}}+x_{n} 
       \left[a_{02} x_{n-1}^2+\left(a_{11}-2\lambda\right) x_{n-1} x_{n+1}+a_{20} x_{n+1}^2\right], 
        \label{eq:Rex2}
    \end{equation}
    Imposing the compatibility condition \eqref{eq:Rcond}
    we obtain $a_{11} = 2\lambda$.
    Using this definition we have the following expressions for
    the functions $M$ and $N$:
    \begin{equation}
        M\left( \xi,\eta \right) = \frac{1}{2}\frac{1}{1-\eta}+a_{20}\eta  \xi^2,
        \quad
        N\left( \xi,\eta \right) = \frac{1}{2}\frac{1}{1-\xi}+a_{02}\xi  \eta^2.
        \label{eq:MNex2}
    \end{equation}
    The closure relation \eqref{eq:closureMN} is then:
    \begin{equation}
            \lambda\pdv{M}{\xi}\left( \xi,\eta \right) -
            \pdv{N}{\eta}\left( \xi,\eta \right)
            =2\xi\eta \left(\lambda a_{20}-a_{02} \right)\equiv0.
        \label{eq:closureMNex2}
    \end{equation}
    This implies that equation \eqref{eq:ex2} with $a_{11}=2\lambda$ is not Lagrangian
    unless $a_{02}=\lambda a_{20}=\lambda\mu$.
    As $M$ and $N$ are defined on the star-shaped domain $D=(1,\infty)\times(1,\infty)$
    we obtain:
    \begin{equation}
        V (x_{n+1},x_{n})
        =
        \frac{1}{2}
        \left[  
            \mu x_{n+1}^2 x_{n}^2-\log(x_{n}-1)-\frac{1}{\lambda}\log(x_{n+1}-1)
        \right].
        \label{eq:Vex2}
    \end{equation}

    Finally we obtained that the one-parameter family of additive fourth-oder
    difference equations:
    \begin{equation}
        \begin{aligned}
       x_{n-1}^2 x_{n-2}&+x_{n+1}^2 x_{n+2}
       +\frac{1}{1-x_{n}}
        \\
       &+x_{n} 
       \left[\mu\left( \lambda x_{n-1}^2+x_{n+1}^2\right) +2 \lambda x_{n-1} x_{n+1}\right]=0,
        \end{aligned}
        \quad
        \lambda^{2}=1,
        \label{eq:ex2c}
    \end{equation}
    can be derived by the following Lagrangian:
    \begin{equation}
        \begin{aligned}
        L_{n} &= \lambda^{-n}\biggl\{ x_{n+1}^{2}x_{n}x_{n+2}
        \biggr.
        \\
        &+\left.
        \frac{1}{2}
        \left[  
            \mu x_{n+1}^2 x_{n}^2-\log(x_{n}-1)-\frac{1}{\lambda}\log(x_{n+1}-1)
    \right]\right\},
        \end{aligned}
    \quad \lambda^{2}=1.
        \label{eq:Lex2}
    \end{equation}
\end{example}

\begin{example}
    \label{ex:lineq}
    In this example we classify the most general variational fourth-order linear
    difference equation:
    \begin{equation}
        x_{n+2} + c_{1} x_{n+1} + c_{0} x_{n} + c_{-1} x_{n-1} + c_{-2} x_{n-2} +b =0.
        \label{eq:lingen4th}
    \end{equation}
    We normalised the equation with respect to the coefficient
    of $x_{n+2}$, which must be different from zero.
    Then we notice that also $c_{-2}\neq0$ in order to
    have a proper fourth-order equation.

    First of all we notice that equation \eqref{eq:lingen4th} is obviously
    denominator free.
    Then $A=1$, $B=c_{-2}$.
    Therefore it follows that $K=1$, $g=1$ and $\lambda^{2}=c_{-2}$. 
    We can then write down equation \eqref{eq:lingen4th} as:
    \begin{equation}
        x_{n+2}+c_{-2} x_{n-2}+
        R\left( x_{n+1},x_{n},x_{n-1} \right)=0,
        \label{eq:lingen4thb}
    \end{equation}
    where the function $R$ is given by:
    \begin{equation}
        R = c_{-1} x_{n-1}+c_{0} x_{n}+c_{1} x_{n+1}+b.
       \label{eq:Rlin}
    \end{equation}
    The compatibility condition \eqref{eq:Rcond} is identically satisfied. 
    Using this definition we have the following expressions for
    the functions $M$ and $N$:
    \begin{equation}
        M\left( \xi,\eta \right) = c_{1} \xi +\frac{c_{0}}{2} \eta+\frac{b}{2} ,
        \quad
        N\left( \xi,\eta \right) = c_{-1} \eta + \frac{c_{0}}{2} \xi+\frac{b}{2}.
        \label{eq:MNlin}
    \end{equation}
    The closure relation \eqref{eq:closureMN} is then:
    \begin{equation}
        c_{-2}^{1/2}\pdv{M}{\xi}\left( \xi,\eta \right) -
            \pdv{N}{\eta}\left( \xi,\eta \right)
            =c_{-2}^{1/2}c_{1}-c_{-1}\equiv0.
        \label{eq:closureMNlin}
    \end{equation}
    This implies that equation \eqref{eq:lingen4th} is Lagrangian
    if and only if $c_{-2}=\left(c_{-1}/c_{1}\right)^{2}$.
    As $M$ and $N$ are defined on the whole $\R^{2}$
    we obtain:
    \begin{equation}
        V (x_{n+1},x_{n})
        =\frac{c_{0}}{4}\left( {\frac {c_{{1}}x_{{n+1}}^{2}}{c_{{-1}}}}+x_{{n}}^{2} \right)
        +\frac{b}{2} \left( x_{{n}}+ {\frac {c_{{1}}x_{{n+1}}}{c_{{-1}}}} \right)
        +c_{{1}}x_{{n}}x_{{n+1}}
        \label{eq:Vlin}
    \end{equation}

    We obtained that the most general Lagrangian fourth-order difference
    equation has the following from:
    \begin{equation}
        x_{n+2} + c_{1} x_{n+1} + c_{0} x_{n} + c_{-1} x_{n-1} 
        + \left(\frac{c_{-1}}{c_{1}}\right)^{2}x_{n-2} +b =0.
        \label{eq:linfin}
    \end{equation}
    and the following Lagrangian:
    \begin{equation}
        L_{n} = \left(\frac{c_{1}}{c_{-1}}\right)^{n}
        \left[
            x_{n}x_{n+2}+
        \frac{c_{0}}{4}\left( {\frac {c_{{1}}x_{{n+1}}^{2}}{c_{{-1}}}}+x_{{n}}^{2} \right)
        +\frac{b}{2} \left( x_{{n}}+ {\frac {c_{{1}}x_{{n+1}}}{c_{{-1}}}} \right)
        +c_{{1}}x_{{n}}x_{{n+1}}
        \right].
        \label{eq:Llin}
    \end{equation}
    Notice that the above Lagrangian becomes independent of $n$ if and
    only if $c_{1}=c_{-1}$.
\end{example}

In the next section we present some results on the integrability
of the additive Lagrangian fourth-order difference equations \eqref{eq:rec4add}.

\section{Integrability results}
\label{sec:int}

In this section we address to the problem of finding some Liouville
integrable examples out of the general family of additive fourth-order 
equations possessing an autonomous Lagrangian, as characterised by 
corollary \ref{cor:structureaut}.
We search for Liouville integrable cases
in the case of fourth-order additive difference equations admitting
an autonomous Lagrangian since Lioville integrability for difference 
equations is is defined for autonomous symplectic structures with autonomous
invariants.
We make an ansatz on the form of the invariant which will allow
us to compare our results with the recent paper \cite{GJTV_class}.
In particular we will show that within the Lagrangian framework we
are able to produce integrable equations imposing only one invariant,
as the second one will be admitted naturally by equation.
Finally, we divide the integrable cases in five canonical forms,
classified up to linear transformations.

\subsection{Additive equations with an invariant multi-affine in $x_{n+1}$ and $x_{n-2}$}
\label{sss:int}

In \cite{GJTV_class} were classified fourth-order difference
equations using the following assumptions:
\begin{enumerate}[label={\Alph*.},font={\bfseries}]
    \item The equation possesses two \emph{symmetric} polynomial invariant
        that is, two invariants 
        $I= I\left( x_{n+1},x_{n},x_{n-1},x_{n-2} \right)$, which are
        polynomial functions and such that:
        \begin{equation}
            I\left( x_{n-2},x_{n-1},x_{n},x_{n+1} \right)
            =
            I\left( x_{n+1},x_{n},x_{n-1},x_{n-2} \right).
            \label{eq:syminv}
        \end{equation}
    \item One invariant, called $I_\text{low}$, is
        such that:
        \begin{equation}
            \deg_{x_{n+1}} I_\text{low} =\deg_{x_{n-2}} I_\text{low} = 1,
            \quad
            \deg_{x_{n}} I_\text{low} =\deg_{x_{n-1}} I_\text{low} = 3,
            \label{eq:dIlow}
        \end{equation}
        and its coefficients interpolates the form of the lowest
        order invariant of the autonomous $\dP_\text{I}^{(2)}$ and $\dP_\text{II}^{(2)}$
        equations (see \cite{GJTV_class,JoshiViallet2017} for details).
    \item One invariant, called $I_\text{high}$, is
        such that:
        \begin{equation}
            \deg_{x_{n+1}} I_\text{high} =\deg_{x_{n-2}} I_\text{high} = 2,
            \quad
            \deg_{x_{n}} I_\text{high} =\deg_{x_{n-1}} I_\text{high} = 4.
            \label{eq:dIhigh}
        \end{equation}
\end{enumerate}

\begin{remark}
    The invariant $I_\text{low}$ is not a affine function, yet it is
    affine in the variables $x_{n+1}$ and $x_{n-2}$.
    A function with this property is said to be \emph{multi-affine}
    with respect to the variables $x_{n+1}$ and $x_{n-2}$.
    \label{rem:multin}
\end{remark}

Within this framework six different equations were derived. 
Some were integrable, some were non-integrable according the algebraic
entropy criterion \cite{BellonViallet1999,FalquiViallet1993,Veselov1992}.
It was proved, following \cite{Gubbiotti_dcov}, that not all integrable 
cases were variational.
The non-variational ones admitted one additional invariant, explaining
integrability in the \emph{na\"ive} sense.
However, variational structure were a key feature in understanding the
integrability of the variational cases.

Now we will discuss the Lioville integrability of
variational fourth-order equations \eqref{eq:rec4addlagr}.
Our final result is stated at the end of this section in Theorem \ref{thm:intmultinlin}.
This result unify the result obtained in \cite{GJTV_class}
and to underline the power of the variational approach.
Our starting point is the existence of a
\emph{single invariant multi-affine with respect to the variables $x_{n+1}$ and $x_{n-2}$},
which is characterised by the following theorem:

\begin{theorem} 
    An additive Lagrangian difference equation of the form \eqref{eq:rec4addlagraut}
    admits a multi-affine invariant with respect to the variables $x_{n+1}$
    and $x_{n-2}$ of the following form:
    \begin{equation}
        \begin{aligned}
            I\left( x_{n+1},x_{n},x_{n-1} ,x_{n-2} \right) &= x_{n+1} P_{1}\left( x_{n},x_{n-1} \right)+x_{n-2} P_{2}\left( x_{n},x_{n-1} \right)
            \\
            &+ x_{n+1}x_{n-2} P_{3}\left( x_{n},x_{n-1} \right) + P_{4}\left( x_{n},x_{n-1} \right),
        \end{aligned}
        \label{eq:int1xx1}
    \end{equation}
    where $P_{i}=P_{i}\left( x_{n},x_{n-1} \right)$ are \emph{a priori}
    arbitrary functions if and only if the following conditions hold
    true:
    \begin{itemize}
        \item The function $g=g\left( \xi \right)$ is a second order
            polynomial in its variables:
            \begin{equation}
                g\left( \xi \right) = A_{1} \xi^{2}+A_{2}\xi+A_{3}.
                \label{eq:gxi}
            \end{equation}
        \item The function $V=V\left( \xi,\eta \right)$ has the following
            form:
            \begin{equation}
                V  = W(\eta)+\frac{A_{1}}{2} \xi^2 \eta^2
                    +A_2 \xi^2 \eta+A_2 \xi \eta^2+A_7 \xi \eta,
                \label{eq:Vxieta}
            \end{equation}
            where the function $W=W\left(\eta  \right)$ is given by integrating:
            \begin{equation}
                W'\left( \eta \right)= 
                \frac{A_2^2 \eta^3+A_2 A_3 \eta^2+A_2 A_7 \eta^2+A_2 A_8 \eta+A_3^2 \eta+A_3 A_8+A_6 \eta+A_5}{A_1 \eta^2+A_2 \eta+A_3},
                \label{eq:Weta}
            \end{equation}
            with initial condition $W\left( 0 \right)=0$.
        \item The functions $P_{i}=P_{i}\left( x_{n},x_{n-1} \right)$ are polynomials 
            in their arguments and have the following form:
            \begin{subequations}
                \begin{align}
                    P_{1}\left( x_{n},x_{n-1} \right) &= -x_{n} 
                    g\left( x_{n} \right)g\left( x_{n-1} \right),
                    \label{eq:P1f}
                    \\
                    P_{2}\left( x_{n},x_{n-1} \right) &= -x_{n-1} 
                    g\left( x_{n} \right)g\left( x_{n-1} \right),
                    \label{eq:P2f}
                    \\
                    P_{3}\left( x_{n},x_{n-1} \right) &= 
                    g\left( x_{n} \right)g\left( x_{n-1} \right),
                    \label{eq:P3f}
                    \\
                    P_{4}\left( x_{n},x_{n-1} \right) &
                    \begin{aligned}[t]
                        &=
                        -x_{n-1}^2 g\left( x_{n} \right) 
                        \left[(A_1 x_{n}+A_2) x_{n-1}+ (2 A_2 x_{n}+A_7) \right]
                        \\
                        &-
                        \left[(A_1 A_3+A_2^2) x_{n}^3+A_2 (2 A_3+A_7) x_{n}^2\right.
                        \\
                        &\phantom{+}\left.+(A_2 A_8+2 A_3^2+A_6) x_{n}+A_3 A_8+A_5\right]
                        x_{n-1}
                        \\
                        &-x_{n} (A_2 A_3 x_{n}^2+A_3 A_7 x_{n}+A_3 A_8+A_5)
                    \end{aligned}
                    \label{eq:P4f}
                \end{align}
                \label{eq:Pf}
            \end{subequations}
    \end{itemize}
    \label{thm:integrability}
\end{theorem}

\begin{proof}
    The proof of this theorem is mainly computational using the explicit form 
    of the invariant \eqref{eq:int1xx1} and that of
    the general additive variational fourth-order difference equation \eqref{eq:rec4addlagr}.
    The starting point is the definition of invariant applied to 
    \eqref{eq:int1xx1}, that is:
    \begin{equation}
        \begin{aligned}
            &x_{n+2} P_{1}\left( x_{n+1},x_{n} \right)+x_{n-1} P_{2}\left( x_{n+1},x_{n} \right)
            + x_{n+2}x_{n-1} P_{3}\left( x_{n+1},x_{n} \right)
            \\
             + &P_{4}\left( x_{n+1},x_{n} \right)
            -x_{n+1} P_{1}\left( x_{n},x_{n-1} \right)-x_{n-2} P_{2}\left( x_{n},x_{n-1} \right)
            \\
            -&x_{n+1}x_{n-2} P_{3}\left( x_{n},x_{n-1} \right) - P_{4}\left( x_{n},x_{n-1} \right)=0.
        \end{aligned}
        \label{eq:fintcond}
    \end{equation}
    After substituting the form of equation \eqref{eq:rec4addlagr} no function
    depends on $x_{n-2}$, so we can take te coefficients with respect to it.
    This yields the following system of functional equations which must be 
    identically satisfied:
    \begin{subequations}
        \begin{gather}
            \begin{aligned}
                &\frac{g(x_{n-1}) P_{1}(x_{n+1}, x_{n})}{g(x_{n+1})}
                +\frac{g(x_{n-1}) x_{n-1} P_{3}(x_{n+1}, x_{n})}{g(x_{n+1})}
                \\
                +&P_{2}(x_{n}, x_{n-1})+x_{n+1} P_{3}(x_{n}, x_{n-1})
                =0,
            \end{aligned}
            \label{eq:fintcondsysa}
            \\
            \begin{aligned}
                &x_{n-1} P_{2}(x_{n+1}, x_{n})
                -x_{n+1} P_{1}(x_{n}, x_{n-1})
                +P_{4}(x_{n+1}, x_{n})
                -P_{4}(x_{n}, x_{n-1})
                \\
                -&\left(g'(x_{n}), x_{n}) x_{n-1} x_{n+1}+\pdv{V(x_{n+1}, x_{n})}{x_{n}}+\pdv{V(x_{n}, x_{n-1})}{x_{n}}\right) \times
                \\
                &\phantom{\times}\left(\frac{P_{1}(x_{n+1}, x_{n})}{g(x_{n+1})}
                + \frac{x_{n-1} P_{3}(x_{n+1}, x_{n})}{g(x_{n+1})}\right)=0.
                \\
            \end{aligned}
            \label{eq:fintcondsysb}
        \end{gather}
        \label{eq:fintcondsys}
    \end{subequations}
    To solve the above equation it is possible to use the following
    strategy:
    \begin{itemize}
        \item Solve either equation \eqref{eq:fintcondsysa} or \eqref{eq:fintcondsysb}
            with respect to one of the unknown functions, 
            e.g. $P_{2}\left( x_{n},x_{n-1} \right)$.
        \item Differentiate with respect to a variable upon which
            such unknown function does not depend, 
            e.g. $x_{n+1}$ in the case of $P_{2}\left( x_{n},x_{n-1} \right)$.
        \item Iterate this procedure until a differential equation containing
            only functions depending on the same set of variables is obtained.
        \item Solve the resulting differential equation and use the previous
            equations as compatibility conditions.
    \end{itemize}
    The outlined procedure is long since several different functions of different variables
    are involved, but only consists of trivial steps.
    For instance, applying this strategy to equation \eqref{eq:fintcondsysa}
    we are able to fix the form of the functions $P_{i}$ in terms of
    the function $g$:
    \begin{subequations}
        \begin{align}
            P_{1}\left( x_{n},x_{n-1} \right) &=
            -g(x_{n}) \left[x_{n} (C_{1}-C_{2} x_{n-1}) g(x_{n-1})-P(x_{n-1})\right]
            \label{eq:P1intermed}
            \\
            P_{2}\left( x_{n},x_{n-1} \right) &=
        -\left[x_{n-1} (C_{2} x_{n}+C_{1}) g(x_{n})+P(x_{n})\right] g(x_{n-1})
            \label{eq:P2intermed}
            \\
            P_{3}\left( x_{n},x_{n-1} \right) &=
            g(x_{n}) \left[(x_{n-1}-x_{n}) C_{2}+C_{1}\right] g(x_{n-1})
            \label{eq:P3intermed}
        \end{align}
        \label{eq:Pintermed}
    \end{subequations}
    where $P=P\left( \xi \right)$ is a still undetermined function and
    $C_{i}$ are constants.
    This values for the functions $P_{i}$ completely solves the first functional
    equation \eqref{eq:fintcondsysa}.
    Inserting this values in the second functional equation \eqref{eq:fintcondsysb}
    we apply the same strategy with respect to the function $V=V\left( x_{n+1},x_{n} \right)$
    and then with respect to $P_{4}\left( x_{n+1},x_{n} \right)$.
    After some steps we find the following equation:
    \begin{equation}
        C_{2} g\left( x_{n} \right) g'\left( x_{n} \right) =0.
        \label{eq:C2s}
    \end{equation}
    This equation implies $C_{2}=0$, as otherwise the function $g=g\left( \xi \right)$
    will be a trivial constant.
    Substituting such values for $C_{2}$ we obtain the following equation for
    $V=V\left( x_{n},x_{n-1} \right)$:
    \begin{equation}
        -C_{1} g''(x_{n-1}) x_{n} + 2 C_{1} g'(x_{n})
        -C_{1} \pdv{V(x_{n}, x_{n-1})}{*{2}{x_{n-1}}, x_{n}}
        +P''(x_{n-1})
        =0.
        \label{eq:Vint}
    \end{equation}
    This last differential equation is readily solved to give the form of
    $V$ in terms of $g$ and $P$:
    \begin{equation}
        \begin{aligned}
            V(x_{n}, x_{n-1})&=
            V_{3}(x_{n})+V_{2}(x_{n}) x_{n-1}+V_{1}(x_{n-1})
            \\
            &-\frac{x_{n}^2}{2} g(x_{n-1})
            +g(x_{n}) x_{n-1}^2+\frac{x_{n}}{C_{1}} P(x_{n-1}).
        \end{aligned}
        \label{eq:Vintsol}
    \end{equation}
    By arbitrariness of $V_{1}\left( x_{n-1} \right)$ we can write 
    $V_{1}\left( x_{n-1} \right)=W\left( x_{n-1} \right) - V_{3}\left( x_{n-1} \right)$
    and remove the total difference $V_{3}\left( x_{n} \right)-V_{3}\left( x_{n-1} \right)$.
    That is, we can write $V\left( x_{n},x_{n-1} \right)$ as:
    \begin{equation}
        \begin{aligned}
            V(x_{n}, x_{n-1})&=
            V_{2}(x_{n}) x_{n-1}+W(x_{n-1})
            \\
            &-\frac{x_{n}^2}{2} g(x_{n-1})
            +g(x_{n}) x_{n-1}^2+\frac{x_{n}}{C_{1}} P(x_{n-1}).
        \end{aligned}
        \label{eq:Vintsol2}
    \end{equation}
    Going back to equation \eqref{eq:fintcondsysb} and removing
    iteratively all the functions depending on $x_{n+1}$ and $x_{n-1}$
    we finally find the following condition on $g$:
    \begin{equation}
        3C_1 g'''(x_{n}) g(x_{n})^3 = 0.
        \label{eq:gcondfin}
    \end{equation}
    As $g$ needs to be non-trivial and $C_{1}\neq0$ from \eqref{eq:Vintsol2}
    we finally obtain that $g$ has to be second order polynomial
    of the form \eqref{eq:gxi}.

    Using the conditions in \eqref{eq:fintcondsysb}
    we find the following expression for the function
    $V_{2}$:
    \begin{equation}
        V_{2}\left( x_{n} \right)=
        \frac{A_{4}}{\sqrt{A_{1} x_{n}^2+A_{2} x_{n}+A_{3}}}
        +\frac{3 A_{2} C_{1} x_{n}^2+2 C_{5} x_{n}+2 C_{6}}{2C_{1}}.
        \label{eq:V2sqrt}
    \end{equation}
    The function $V_{2}\left( x_{n} \right)$ appears to be \emph{algebraic} in $x_{n}$.
    However, substituting back in order to check the compatibility
    conditions we obtain $A_{4}=0$.
    Therefore, no algebraic term is left.

    The above computations produce rahter cumbersome expression for 
    $P_{4}\left( x_{n},x_{n-1} \right)$, which we will not reproduce here.
    However, we notice that this final form of $P_{4}\left( x_{n},x_{n-1} \right)$ 
    yield the following condition for the function $W=W\left( \eta \right)$:
    \begin{equation}
        W'\left( \eta \right) =
        \frac{1}{C_{1}}
        \frac{C_{1} (A_{2}^2 -A_{1} A_{3}) \eta^3-(A_{1} C_{6}-A_{2} C_{5}) \eta^2+A_{6} C_{1} \eta+A_{5} C_{1}}{A_{1} \eta^2+A_{2} \eta+A_{3}}
        \label{eq:Wdiff0}
    \end{equation}
    Since $C_{1}\neq 0$ we perform the scaling $C_{5}=A_{7}C_{1}$ and
    $C_{6}=A_{8}C_{1}$.
    This finally yield the expression \eqref{eq:Weta} for $W\left( \eta \right)$ 
    and \eqref{eq:P4f} for $P_{4}\left( x_{n},x_{n-1} \right)$ and concludes the proof.
\end{proof}

\begin{remark}
    The explicit expression of the additive Lagrangian difference equations 
    with one integral of the form \eqref{eq:int1xx1} is given by:
    \begin{equation}
        \begin{aligned}
            &(A_{1} x_{n-1}^2+A_{2} x_{n-1}+A_{3}) x_{n-2}
            +(A_{1} x_{n+1}^2+A_{2} x_{n+1}+A_{3}) x_{n+2}
            \\
            +&(A_{1} x_{n}+A_{2}) \left(x_{n+1}^2+x_{n-1}^2\right)+(2 A_{1} x_{n}+A_{2}) x_{n-1} x_{n+1}
            \\
            +&(2 A_{2} x_{n}+A_{7}) \left(x_{n+1} + x_{n-1}\right)
            \\
            +&\frac{A_{2}^2 x_{n}^3+\left(A_{2} A_{3} +A_{2} A_{7}\right) x_{n}^2
                +\left(A_{2} A_{8} +A_{3}^2 +A_{3} A_{8}+A_{6}\right) x_{n}+A_{5}}{A_{1} x_{n}^2+A_{2} x_{n}+A_{3}}=0.
        \end{aligned}
        \label{eq:addint}
    \end{equation}
    We choose to not present the explicit form of the Lagrangian
    for equation \eqref{eq:addint} yet, since it depends on the 
    functional form of the solution of the differential equation \eqref{eq:Weta}.
    Such solution is different depending on the values of the parameters $A_{i}$, 
    and it is impossible to write down in full generality.
    We will present the explicit Lagrangians later when we will discuss
    the canonical forms of equation \eqref{eq:addint}.
    \label{rem:completeform}
\end{remark}

By direct inspection it is possible to prove that equation
\eqref{eq:addint} possess a second invariant of higher order,
which has the following expression:
\begin{equation}
    \begin{aligned}   
        J &= -A_{1} g\left( x_{n} \right)^{2}g\left( x_{n-1} \right)^{2} 
        \left( x_{n+1}^{2}+x_{n-2}^{2} \right)
        \\
        &-g\left( x_{n} \right)g\left( x_{n-1} \right) 
        \left( 
            \begin{gathered}
                2 A_1^2 x_{n} x_{n-1}+A_1 A_2 x_{n}+A_1 A_2 x_{n-1}
                \\
                +5 A_1 A_3+2 A_1 A_7+2 A_2^2
            \end{gathered}
         \right) 
        x_{n+1}x_{n-2}
        \\
        &-g\left( x_{n} \right) Q\left( x_{n-1},x_{n} \right)x_{n+1}
        -g\left( x_{n-1} \right) Q\left( x_{n},x_{n-1} \right) x_{n}
        +R\left( x_{n},x_{n-1} \right),
    \end{aligned}
    \label{eq:Jint}
\end{equation}
where $Q$ and $R$ are two polynomial with the following
expression:
\begin{subequations}
    \begin{align}
            Q\left( \xi,\eta \right) & 
            \begin{aligned}[t]
            &=
        2 A_1^3 \eta^2 \xi^3+4 A_1^2 A_2 \eta^2 \xi^2+3 A_1^2 A_2 \eta \xi^3+2 A_1^2 A_3 \eta^2 \xi
        \\
        &-5 A_1^2 A_3 \eta \xi^2+2 A_1 A_2^2 \eta^2 \xi+A_1 A_2^2 \eta \xi^2
        +A_1 A_2^2 \xi^3+2 A_1 A_2 A_3 \eta^2
        \\
        &-2 A_1 A_2 A_3 \eta \xi+A_1 A_2 A_3 \xi^2+A_1 A_2 A_7 \xi^2-2 A_2^3 \eta \xi+A_1 A_2 A_8 \xi
        \\
        &-5 A_1 A_3^2 \eta+A_1 A_3^2 \xi-2 A_2^2 A_3 \eta+A_1 A_3 A_8+A_1 A_6 \xi+A_1 A_5
            \end{aligned}
            \label{eq:Qpol}
            \\
            R\left( \xi,\eta \right) & 
            \begin{aligned}[t]
                &=
                -A_1 A_2^2 A_3 \xi^4-A_1^4 \xi^4 \eta^4-3 A_1^3 A_2 \xi^4 \eta^3
                \\
                &-A_1^2 (A_1 A_3+3 A_2^2) \eta^2 \xi^4
                \\
                &-A_2 A_1 (2 A_1 A_3+A_2^2) \eta \xi^4+A_2 A_3 (5 A_1 A_3+2 A_2^2) \xi^3
                \\
                &-3 A_1^3 A_2 \xi^3 \eta^4+A_1^2 (5 A_1 A_3-4 A_2^2) \eta^3 \xi^3
                \\
                &+A_2 A_1 (7 A_1 A_3+A_2^2) \eta^2 \xi^3+(5 A_1^2 A_3^2+4 A_1 A_2^2 A_3+2 A_2^4) \eta \xi^3
                \\
                &+\left(
                    \begin{gathered}
                        5 A_1 A_3^2 A_7+A_1 A_3 A_7^2+2 A_2^2 A_3 A_7
                        \\
                        -A_1 A_2 A_5-A_1 A_2 A_3 A_8-A_1 A_3^3
                    \end{gathered}
                \right) \xi^2
                \\
                &-A_1^2 (A_1 A_3+3 A_2^2) \eta^4 \xi^2+A_2 A_1 (7 A_1 A_3+A_2^2) \eta^3 \xi^2
                \\
                &+\left(
                    \begin{gathered}
                    5 A_1^2 A_3 A_7+A_1^2 A_7^2+6 A_1 A_2^2 A_3+2 A_1 A_2^2 A_7
                    \\
                    +4 A_2^4-A_1^2 A_6-A_1^2 A_2 A_8-3 A_1^2 A_3^2
                    \end{gathered}
                \right) \eta^2 \xi^2
                \\
                &+\left(
                    \begin{gathered}
                        7 A_1 A_2 A_3^2+6 A_1 A_2 A_3 A_7+A_1 A_2 A_7^2+4 A_2^3 A_3+2 A_2^3 A_7
                        \\
                        -A_1^2 A_5-A_1 A_2 A_6-A_1^2 A_3 A_8-A_1 A_2^2 A_8
                    \end{gathered}
            \right) \eta \xi^2
                \\
                &+(A_3 A_8+A_5) (4 A_1 A_3+A_1 A_7+2 A_2^2) \xi-A_2 A_1 (2 A_1 A_3+A_2^2) \eta^4 \xi
                \\
                &+(5 A_1^2 A_3^2+4 A_1 A_2^2 A_3+2 A_2^4) \eta^3 \xi
                \\
                &+\left(
                    \begin{gathered}
                        7 A_1 A_2 A_3^2+6 A_1 A_2 A_3 A_7+A_1 A_2 A_7^2+4 A_2^3 A_3
                        \\
                        +2 A_2^3 A_7-A_1^2 A_5-A_1 A_2 A_6-A_1^2 A_3 A_8-A_1 A_2^2 A_8
                    \end{gathered}
                \right) \eta^2 \xi
                \\
                &+\left(
                    \begin{gathered}
                        3 A_1 A_2 A_3 A_8+A_1 A_2 A_7 A_8+10 A_1 A_3^3
                        \\
                        +A_1 A_3^2 A_7
                        +2 A_2^3 A_8+4 A_2^2 A_3^2-2 A_1 A_2 A_5
                        \\
                        +5 A_1 A_3 A_6 +A_1 A_6 A_7+2 A_2^2 A_6
                    \end{gathered}
                    \right) \eta \xi
                \\
                &-A_1 A_2^2 A_3 \eta^4+A_2 A_3 (5 A_1 A_3+2 A_2^2) \eta^3
                \\
                &+\left(
                    \begin{gathered}
                        5 A_1 A_3^2 A_7+A_1 A_3 A_7^2+2 A_2^2 A_3 A_7
                        \\
                        -A_1 A_2 A_5-A_1 A_2 A_3 A_8-A_1 A_3^3
                    \end{gathered}
                \right) \eta^2
                \\
                &+(A_3 A_8+A_5) (4 A_1 A_3+A_1 A_7+2 A_2^2) \eta.
            \end{aligned}
    \label{eq:Rpol}
    \end{align}
    \label{eq:RQpol}
\end{subequations}
For general values of the parameters $A_{i}$ computing the
rank of the Jacobian of two invariants \eqref{eq:int1xx1} and
\eqref{eq:Jint} it is possible to prove that the two invariants
are functionally independent.
This means that equation \eqref{eq:addint} is candidate to be
an integrable equation, as it possesses two independent invariants
and it is variational by construction.
To complete the proof of integrability of equation \eqref{eq:addint}
we need to show that the two invariants \eqref{eq:int1xx1} and
\eqref{eq:Jint} are Poisson commuting with respect to the Poisson 
structure defined by the Lagrangian.
We defer this part of the proof to subsection \ref{sss:canonical}.

\subsection{Admissible transformations}
In the previous subsection we presented a general equation \eqref{eq:addint},
possessing two invariants with a given form.
We also provided formula \eqref{eq:Vxieta} which, up to integration,
provides the Lagrangian and hence the symplectic structure for
such equation.
The functional form of the Lagrangian derived from \eqref{eq:Vxieta} 
depends on the parameters $A_{i}$.
To simplify this expression we introduce a simple class of admissible transformations,
allows us to enumerate a \emph{finite number} of subcases 
of equation \eqref{eq:addint}.
These subcases will depend on fewer parameters and will have a
definite functional form of the Lagrangian.
Consider then the following:
\begin{lemma}
    An additive difference equation \eqref{eq:rec4add} 
    is \emph{form invariant} under linear point transformation
    \begin{equation}
        x_{n} = a X_{n} + b.
        \label{eq:lingen}
    \end{equation}
    That is an additive difference equation \eqref{eq:rec4add}
    under the transformation \eqref{eq:lingen} is transformed
    into another additive difference equation.
    Moreover, also the Lagrangian \eqref{eq:rec4addL} for an 
    additive difference equation \eqref{eq:rec4add} is form
    invariant under the linear point transformation \eqref{eq:lingen}.
    \label{lem:admissibletransf}
\end{lemma}

\begin{proof}
    The linear transformation \eqref{eq:lingen}
    maps an additive difference equation \eqref{eq:rec4add} 
    with functions $f$ and $h$ into:
    \begin{equation}
        X_{n+2} = \tilde{f}\left( X_{n+1},X_{n},X_{n-1} \right) X_{n-2}+
        \tilde{h}\left( X_{n+1},X_{n},X_{n-1} \right),
        \label{eq:lingenappl}
    \end{equation}
    where:
    \begin{subequations}
        \begin{align}
            \tilde{f}\left( X_{n+1},X_{n},X_{n-1} \right) &=
            f \left( a X_{n+1} +b,aX_{n}+b,aX_{n-1}+b\right),
            \\
            \tilde{h}\left( X_{n+1},X_{n},X_{n-1} \right) &
            \begin{aligned}[t]
                &=
                \frac{b}{a}\left[\tilde{f}(X_{n+1}, X_{n}, X_{n-1})-1\right]
                \\
                &+\frac{1}{a}h(a X_{n+1}+b, a X_{n}+b, a X_{n-1}+b).
            \end{aligned}
        \end{align}
        \label{eq:ftht}
    \end{subequations}
    As equation \eqref{eq:lingenappl} is still an additive
    difference equation we proved the first part of the statement.
    
    In the same way applying the linear point 
    transformation \eqref{eq:lingen} to the Lagrangian
    \eqref{eq:rec4addL} we obtain:
    \begin{equation}
        \begin{aligned}
            L_{n} &= \lambda^{-n}
            \left[
            g\left( aX_{n+1}+b \right)\left( a X_{n}+b \right)\left( a X_{n+2}+b \right)\right.
            \\
            &\left.\quad\quad\quad\quad
            +V\left( aX_{n+1}+b,aX_{n}+b \right)\right]
            \\
            &=
            \lambda^{-n}\left[
            a^{2} g\left( aX_{n+1}+b \right) X_{n}X_{n+2}
            +a b g\left( aX_{n+1}+b \right) \left(  X_{n}+X_{n+2}\right)
            \right.
            \\
            &\left.\quad+b^{2} g\left( aX_{n+1}+b \right)
            +V\left( aX_{n+1}+b,aX_{n}+b \right)\right]
            \\
        \end{aligned}
        \label{eq:Llintransf}
    \end{equation}
    Noting that:
    \begin{equation}
        \lambda^{-n} g\left( X_{n+1} \right)X_{n+2} =
        \lambda^{-n+1} g\left( X_{n} \right)X_{n+1}
        + \left( T_{n} -\Id \right)\left[ \lambda^{-n+1} g\left( X_{n} \right)X_{n+1} \right]
        \label{eq:gnptogn}
    \end{equation}
    we obtain:
    \begin{equation}
        L_{n} 
        \begin{aligned}[t]
        &\lequiv \lambda^{-n}
            \left\{
             a^{2} g\left( aX_{n+1}+b \right) X_{n}X_{n+2}
         +b^{2} g\left( aX_{n+1}+b \right)\right.
            \\
            &\quad+a b  \left[  g\left( aX_{n+1}+b \right)X_{n}
            + \lambda g\left( aX_{n}+b \right) X_{n+1}\right]
            \\
            &\left.\quad+V\left( aX_{n+1}+b,aX_{n}+b \right)\right\}.
        \end{aligned}
        \label{eq:Llintransf_fin}
    \end{equation}
    That is defining:
    \begin{subequations}
        \begin{align}
            g\left( X_{n} \right) & = a^{2} g\left( aX_{n+1}+b \right), 
            \\
            V\left( X_{n+1},X_{n} \right) &
            \begin{aligned}[t]
                    &=
                    a b  \left[  g\left( aX_{n+1}+b \right)X_{n}
                    +\lambda g\left( aX_{n}+b \right) X_{n+1}\right]
                    \\
                    &+b^{2} g\left( aX_{n+1}+b \right)+V\left( aX_{n+1}+b,aX_{n}+b \right)
                \end{aligned}
        \end{align}
        \label{eq:gVlin}
    \end{subequations}
    we obtain the second part of the statement.
\end{proof}

From lemma  \ref{lem:admissibletransf} follows that we can classify additive
fourth-order difference equation up to equivalence with respect to
linear transformations.
So, we define an \emph{admissible transformation} to be a linear
point transformation \eqref{eq:lingen}.

\subsection{Canonical forms}
\label{sss:canonical}

Consider now the equation \eqref{eq:addint}.
This equation depends on a polynomial $g\left( \xi \right)$ \eqref{eq:gxi},
which in the general case has degree two.
Depending on the values of the coefficients $A_{1}$, $A_{2}$ and
$A_{3}$, which we assume to be real, the polynomial $g\left( \xi \right)$ \eqref{eq:gxi}
can be of the following five forms:
\begin{description}
    \item[Case 1] $\deg g = 2$ and it has two real independent solutions $x_{1}$ and 
        $x_{2}$.
    \item[Case 2] $\deg g = 2$ and it has one solutions $x_{0}$ of multiplicity two.
    \item[Case 3] $\deg g = 2$ and it has two complex conjugate solutions $x_{0}$ and 
        $x_{0}^{*}$.
    \item[Case 4] $\deg g = 1$.
    \item[Case 5] $\deg g = 0$.
\end{description}

We will now consider explicitly these five possibility and
show, using the form invariance with respect to linear transformations
\eqref{eq:lingen} that they give raise to five different
\emph{canonical forms} of equation \eqref{eq:addint}.
That is,
an equation of the form \eqref{eq:addint}
for a specific choice of the parameters reduces to one
of these five using the appropriate linear transformation
and reparametrisation we will show.
Finally, we note that these canonical forms will show the true
number of independent parameters and will be helpful to 
compute the continuum limits in section \ref{sec:contlim}.

\vspace{12pt}


\paragraph{Case 1}
If $\deg g = 2$ and it has two real independent solutions $x_{1}$ and $x_{2}$, 
we can write $g$ in the following way:
\begin{equation}
    g\left( \xi \right) = \kappa \left( \xi-x_{1} \right)\left( \xi-x_{2} \right).
    \label{eq:g1}
\end{equation}
Without loss of generality we assume $x_{2}>x_{1}$
Applying the linear transformation
\begin{equation}
    \xi  = \frac{x_{2}-x_{1}}{2} \Xi + \frac{x_{1}+x_{2}}{2}, 
    \label{eq:l1}
\end{equation}
we bring the two roots of the polynomial \eqref{eq:g1} into
the canonical values $\Xi=-1,1$.

Defining the linear point transformation:
\begin{equation}
    x_{n}  = \frac{x_{2}-x_{1}}{2} X_{n} + \frac{x_{1}+x_{2}}{2}, 
    \label{eq:L1}
\end{equation}
we bring equation \eqref{eq:addint} when $g$ is given by \eqref{eq:g1}
into the following equation:
\begin{equation}
    \begin{aligned}
        (X_{n+1}^{2}-1) X_{n+2} &+ (X_{n-1}^{2}-1) X_{n-2}
        +X_{n} \left( X_{n+1}+X_{n-1} \right)^{2}
        \\
        &+\gamma\left( X_{n+1}+X_{n-1}\right) 
        + \frac{\alpha X_{n}+\beta}{X_{n}^{2}-1}=0.
    \end{aligned}
    \label{eq:can1}
\end{equation}
This is the first canonical form.

The first canonical form \eqref{eq:can1} depends on three 
parameters $\alpha$, $\beta$ and $\gamma$ which 
are related to the old ones through the following definitions:
\begin{subequations}
    \begin{align} 
        A_5 & 
        \begin{aligned}[t]    
            &=- \left[  \left(  \left( \beta+\gamma \right) \frac{\delta^{5}}{32}
                + \left( {  \alpha}+{  \beta}/4+2 \right) \frac{  x_2 \delta^{4}}{4}
            + \left(\alpha+6 \right)\frac{3x_2^{2}{\delta}^{3}}{4}\right.\right.
                \\
                &\left.\left.\quad\quad+ \left( {  \alpha}+26 \right) \frac{x_2^{3}\delta^{2}}{2}
                +15  x_2^{4}\delta+6 x_2^{5} \right) \kappa
            +A_8  x_2  \left( \delta+x_2 \right)  \right] \kappa,
        \end{aligned}
        \label{eq:par2a}
        \\
        A_6 &
        \begin{aligned}[t]
        &=
        2 \kappa  \left\{  
            \left[  \left( \frac{\alpha}{8}+\frac{\beta}{32}+\frac{1}{4} \right) {\delta}^{4}
                +\frac{x_2}{2}  \left( \alpha+6 \right) {\delta}^{3}
                +\frac{x_2^{2}}{2} \left( \alpha+21 \right) {\delta}^{2}
        \right.\right.
                \\
                &\left.\left.\quad\quad\quad
            +15 \delta x_2^{3}+\frac{15x_2^{4}}{2} \right] \kappa
            + A_8 \left( x_2+\delta/2 \right)  \right\},
        \end{aligned}
        \label{eq:par2b}
        \\
        A_7 &=  
        \frac{\kappa}{4}  \left( \alpha+6 \right) \delta^{2}
        +6 \kappa \delta x_2+6 \kappa x_2^{2},
        \label{eq:par2c}
    \end{align}
    \label{eq:par2}
\end{subequations}
where $\delta=x_{2}-x_{1}$.

The first  canonical form \eqref{eq:can1} is, up to change of the parameters, 
the autonomous the second member of the discrete $P_\text{II}$ hierarchy, the 
$\dP_{II}^{(2)}$ equation.
The $\dP_{II}^{(2)}$ equation was presented in \cite{CresswellJoshi1999} and
the integrability properties with respect to invariants and growth
of the degrees 
\cite{FalquiViallet1993,BellonViallet1999,HietarintaViallet1997,
GrammaticosHalburdRamaniViallet2009,GubbiottiASIDE16}
of this equation were investigated in \cite{JoshiViallet2017}.
This equation reappeared later in the classification 
given in \cite{GJTV_class}, see subsection \ref{sss:int}.
In \cite{GJTV_class} the growth properties of equation \eqref{eq:can1}
were explained proving that such equation is Liouville integrable.
Its Lagrangian, found with the method of \cite{Gubbiotti_dcov},
and the associated symplectic structure were presented.
For sake of completeness here we are going to present again such
properties.

The Lagrangian of the first canonical form \eqref{eq:can1} is the following:
\begin{equation}
    \begin{aligned}
        L_{1} &=
        (X_{n+1}^{2}-1) X_{n+2} X_{n}
        +\frac{1}{2}X_{n}^{2}X_{n+1}^2
        +\gamma X_{n}X_{n+1}
        \\
        &+\frac{\alpha}{2} \log\left( X_{n}^{2}-1 \right)
        +\frac{\beta}{2}\log\left( \frac{X_{n}-1}{X_{n}+1} \right).
    \end{aligned}
    \label{eq:Lagr1}
\end{equation}
The invariants of the first canonical form \eqref{eq:can1} are obtained
from formul\ae\ \eqref{eq:int1xx1} and \eqref{eq:Jint} after performing
all the appropriate substitutions.
As they are quite cumbersome, we omit to write down their explicit expression
here.
However, to complete the explicit proof of integrability of the first canonical
form we present the form of the sympectic structure obtained from
the Lagrangian \eqref{eq:Lagr1}.
Such Poisson structure has the following
non-zero brackets:
\begin{subequations}
    \begin{align}
        \left\{ X_{n+1},X_{n-1} \right\} &= 
        -\frac{1}{X_n^{2}-1},
        \\
        \left\{ X_{n+1},X_{n-2} \right\} &= 
        \frac{2 X_{n} X_{n-1}+2 X_{n} X_{n+1}+2 X_{n-2} X_{n-1}+\gamma}{%
        X_{n}^2 X_{n-1}^2-X_{n}^2-X_{n-1}^2+1}
        \\
        \left\{ X_{n},X_{n-2} \right\} &= 
        -\frac{1}{X_{n-1}^{2}-1},
    \end{align}
    \label{eq:J1}
\end{subequations}
Using such Poisson structure it is possible to prove that 
\emph{mutatis mudandis} the
two invariants \eqref{eq:int1xx1} and \eqref{eq:Jint} 
are commuting.
This ends the proof of the integrability of the first canonical canonical form
\eqref{eq:can1}.

\paragraph{Case 2}
If $\deg g = 2$ and it has one solution $x_{0}$ of multiplicity two, 
we can write $g$ in the following way:
\begin{equation}
    g\left( \xi \right) = \kappa \left( \xi-x_{0} \right)^{2}.
    \label{eq:g2}
\end{equation}
Applying the linear transformation
\begin{equation}
    \xi  =  \Xi + x_{0}, 
    \label{eq:l3}
\end{equation}
we bring the two roots of the polynomial \eqref{eq:g2} into
the canonical value $\Xi=0$.

Defining the linear point transformation:
\begin{equation}
    x_{n}  = X_{n} + x_{0}, 
    \label{eq:L3}
\end{equation}
we bring equation \eqref{eq:addint} when $g$ is given by \eqref{eq:g2}
into the following equation:
\begin{equation}
    \begin{aligned}
        X_{n+1}^{2} X_{n+2} &+ X_{n-1}^{2} X_{n-2}
        +X_{n} \left( X_{n+1}+X_{n-1} \right)^{2}
        \\
        &+\gamma\left( X_{n+1}+X_{n-1}\right) 
        + \frac{\alpha}{X_{n}^{2}}+\frac{\beta}{X_{n}}=0.
    \end{aligned}
    \label{eq:can2}
\end{equation}
This is the second canonical form.

The second canonical form \eqref{eq:can2} depends on three 
parameters $\alpha$, $\beta$ and $\gamma$ which 
are related to the old ones through the following definitions:
\begin{subequations}
    \begin{align} 
        A_5 & = \kappa (\alpha \kappa-6 \kappa x_{0}^5
        -2 \gamma \kappa x_{0}^3-\beta \kappa x_{0}-A_8 x_{0}^2),
        \label{eq:par3a}
        \\
        A_6 &= \kappa (15 \kappa x_{0}^4+4 \gamma \kappa x_{0}^2+\beta \kappa+2 A_8 x_{0}),
        \label{eq:par3b}
        \\
        A_7 &= \kappa (6 x_{0}^2+\gamma) . 
        \label{eq:par3c}
    \end{align}
    \label{eq:par3}
\end{subequations}

The second canonical form \eqref{eq:can2} is, up to change of the parameters,
equation (P.v) appearing in the classification of fourth-order difference
equations with two invariants of a given form presented in \cite{GJTV_class}.
In \cite{GJTV_class} the growth properties of equation \eqref{eq:can2}
were explained proving that such equation is Liouville integrable.
Its Lagrangian, found with the method of \cite{Gubbiotti_dcov},
and the associated symplectic structure were presented.
For sake of completeness here we are going to present again such
properties.

The Lagrangian of the second canonical form \eqref{eq:can2} is the following:
\begin{equation}
    \begin{aligned}
        L_{2} &=
        X_{n+1}^{2}X_{n} X_{n+2}
        +\frac{1}{2}X_{n}^{2}X_{n+1}^2
        +\gamma X_{n}X_{n+1}
        -\frac{\alpha}{X_{n}}+\beta\log \left(X_{n}\right).
    \end{aligned}
    \label{eq:Lagr2}
\end{equation}
The invariants of the second canonical form \eqref{eq:can2} are obtained
from formul\ae\ \eqref{eq:int1xx1} and \eqref{eq:Jint} after performing
all the appropriate substitutions.
As they are quite cumbersome, we omit to write down their explicit expression
here.
However, to complete the explicit proof of integrability of the first canonical
form we present the form of the sympectic structure obtained from
the Lagrangian \eqref{eq:Lagr2}.
Such Poisson structure has the following
non-zero brackets:
\begin{subequations}
    \begin{align}
        \left\{ X_{n+1},X_{n-1} \right\} &= 
        -\frac{1}{X_n^{2}},
        \\
        \left\{ X_{n+1},X_{n-2} \right\} &= 
        \frac{2 X_{n} X_{n-1}+2 X_{n} X_{n+1}+2 X_{n-2} X_{n-1}+\gamma}{%
        X_{n}^2 X_{n-1}^2}
        \\
        \left\{ X_{n},X_{n-2} \right\} &= 
        -\frac{1}{X_{n-1}^{2}},
    \end{align}
    \label{eq:J2}
\end{subequations}
Using such Poisson structure it is possible to prove that 
\emph{mutatis mudandis} the
two invariants \eqref{eq:int1xx1} and \eqref{eq:Jint} 
are commuting.
This ends the proof of the integrability of the second canonical canonical form
\eqref{eq:can2}.

\paragraph{Case 3}
If $\deg g = 2$ and it has two complex conjugate solutions $x_{0}$ and $x_{0}^{*}$, 
we can write $g$ in the following way:
\begin{equation}
    g\left( \xi \right) = \kappa \left( \xi-2\mu x +\mu^{2}+\nu^{2} \right),
    \label{eq:g3}
\end{equation}
where $x_{0} = \mu + \imath \nu$.
Applying the linear transformation
\begin{equation}
    \xi  =  \nu\Xi +\mu, 
    \label{eq:l4}
\end{equation}
we bring the two roots of the polynomial \eqref{eq:g3} into
the canonical values $\Xi=\pm \imath$.

Defining the linear point transformation:
\begin{equation}
    x_{n}  = \nu X_{n} + \mu, 
    \label{eq:L4}
\end{equation}
we bring equation \eqref{eq:addint} when $g$ is given by \eqref{eq:g3}
into the following equation:
\begin{equation}
    \begin{aligned}
        \left(X_{n+1}^{2}+1\right) X_{n+2} &+\left( X_{n-1}^{2}+1\right) X_{n-2}
        +X_{n} \left( X_{n+1}+X_{n-1} \right)^{2}
        \\
        &+\gamma\left( X_{n+1}+X_{n-1}\right) 
        + \frac{\alpha + \beta X_{n}}{X_{n}^{2}+1}=0.
    \end{aligned}
    \label{eq:can3}
\end{equation}
This is the third canonical form.

The third canonical form \eqref{eq:can3} depends on three 
parameters $\alpha$, $\beta$ and $\gamma$ which 
are related to the old ones through the following definitions:
\begin{subequations}
    \begin{align} 
        A_5 & 
        \begin{aligned}[t]
            &= \kappa^{2}
            \left[\alpha \nu^5-\nu^4 \left(\beta \mu +2 \mu +2 \gamma \mu\right)  
            -2\mu^3 \nu^{2}\left( \gamma +4 \right)-6 \mu^5\right]
        \\
        &-\kappa A_8 \left(\mu^2+ \nu^2\right)
        \end{aligned}
        \label{eq:par4a}
        \\
        A_6 &= \kappa (\beta \kappa \nu^4+4 \gamma \kappa \mu^2 \nu^2+15 \kappa \mu^4+6 \kappa \mu^2 \nu^2-\kappa \nu^4+2 A_8 \mu)
        \label{eq:par4b}
        \\
        A_7 &= \kappa (\gamma \nu^2+6 \mu^2) 
        \label{eq:par4c}
    \end{align}
    \label{eq:par4}
\end{subequations}

The third  canonical form \eqref{eq:can3} is, connected to the first
one \eqref{eq:can1}, if we allow complex changes of variables and
complexify the parameters:
\begin{equation}
    X_{n} \leftrightarrow \imath X_{n},
    \quad
    (\alpha ,\beta,\gamma) \leftrightarrow (\imath \beta,\alpha,-\gamma).
    \label{eq:can3tocan2}
\end{equation}
Therefore the third canonical form \eqref{eq:can3} is a different avatar
of the autonomous $\dP_\text{II}^{(2)}$ equation.
We choose to consider it as different equation, because 
the explicit expression of the Lagrangian for equation
\eqref{eq:can3} is different with respect to the one of equation \eqref{eq:can1}.
Moreover, in section \ref{sec:contlim}, we will show that
the continuum limit of the third canonical form \eqref{eq:can3}
is different from the continuum limit of the second canonical form
\eqref{eq:can2}.

The Lagrangian of the third canonical form \eqref{eq:can3} is the following:
\begin{equation}
    \begin{aligned}
        L_{3} &=
        \left(X_{n+1}^{2}+1\right)X_{n} X_{n+2}
        +\frac{1}{2}X_{n}^{2}X_{n+1}^2
        +\gamma X_{n}X_{n+1}
        \\
        &+\frac{\alpha}{2}\arctan\left(X_{n}\right)+
        \beta\log \left(X_{n}^{2}+1\right).
    \end{aligned}
    \label{eq:Lagr3}
\end{equation}
The invariants of the third canonical form \eqref{eq:can3} are obtained
from formul\ae\ \eqref{eq:int1xx1} and \eqref{eq:Jint} after performing
all the appropriate substitutions.
As they are quite cumbersome, we omit to write down their explicit expression
here.
However, to complete the explicit proof of integrability of the first canonical
form we present the form of the sympectic structure obtained from
the Lagrangian \eqref{eq:Lagr3}.
Such Poisson structure has the following
non-zero brackets:
\begin{subequations}
    \begin{align}
        \left\{ X_{n+1},X_{n-1} \right\} &= 
        -\frac{1}{X_n^{2}+1},
        \\
        \left\{ X_{n+1},X_{n-2} \right\} &= 
        \frac{2 X_{n} X_{n-1}+2 X_{n} X_{n+1}+2 X_{n-2} X_{n-1}+\gamma}{%
        X_{n}^2 X_{n-1}^2+X_{n}^{2}+X_{n+1}^{2}+1}
        \\
        \left\{ X_{n},X_{n-2} \right\} &= 
        -\frac{1}{X_{n-1}^{2}+1},
    \end{align}
    \label{eq:J3}
\end{subequations}
Using such Poisson structure it is possible to prove that 
\emph{mutatis mudandis} the
two invariants \eqref{eq:int1xx1} and \eqref{eq:Jint} 
are commuting.
This ends the proof of the integrability of the third canonical canonical form
\eqref{eq:can3}.

\paragraph{Case 4}
If $\deg g = 1$ we can write $g$ in the following way:
\begin{equation}
    g\left( \xi \right) = \mu \xi +\nu,
    \label{eq:g4}
\end{equation}
where we assume $\mu\neq0$.
Applying the linear transformation
\begin{equation}
    \xi  =  \frac{\Xi -\nu}{\mu}, 
    \label{eq:l5}
\end{equation}
we bring the root of the polynomial \eqref{eq:g4} into
the canonical value $\Xi=0$.

Defining the linear point transformation:
\begin{equation}
    x_{n}  = \frac{X_{n} -\nu}{\mu}, 
    \label{eq:L5}
\end{equation}
we bring equation \eqref{eq:addint} when $g$ is given by \eqref{eq:g4}
into the following equation:
\begin{equation}
    \begin{aligned}
        X_{n+1} X_{n+2} &+ X_{n-1} X_{n-2}
        +X_{n} \left(X_{n}+ 2X_{n+1}+2X_{n-1} \right)
        +\left( X_{n+1}+X_{n-1} \right)^{2}
        \\
        &-X_{n+1}X_{n-1}
        +\gamma\left(X_{n}+ X_{n+1}+X_{n-1}\right) 
        +\frac{\alpha}{X_{n}}+\beta=0.
    \end{aligned}
    \label{eq:can4}
\end{equation}
This is the fourth canonical form.

The fourth canonical form \eqref{eq:can4} depends on three 
parameters $\alpha$, $\beta$ and $\gamma$ which 
are related to the old ones through the following definitions:
\begin{subequations}
    \begin{align} 
        A_5 & =
        \frac{\nu^3-(6 \nu+\gamma) \nu^2+(4 \gamma \nu-A_8 \mu+15 \nu^2+\beta) \nu+\alpha}{\mu},
        \label{eq:par5a}
        \\
        A_6 &= 4 \gamma \nu-A_8 \mu+15 \nu^2+\beta
        \label{eq:par5b}
        \\
        A_7 &= 6 \nu+\gamma 
        \label{eq:par5c}
    \end{align}
    \label{eq:par5}
\end{subequations}

The fourth  canonical form \eqref{eq:can4} is, up to change of the parameters, 
the autonomous second member of the discrete $P_\text{I}$ hierarchy, the $\dP_\text{I}^{(2)}$ equation.
The $\dP_\text{I}^{(2)}$ equation was presented in \cite{CresswellJoshi1999} and
the integrability properties with respect to invariants and growth
of the degrees of this equation were investigated in \cite{JoshiViallet2017}.
Alongside with equations \eqref{eq:can1} and \eqref{eq:can2}
this equation reappeared later in the classification 
given in \cite{GJTV_class}, see subsection \ref{sss:int}.
In \cite{GJTV_class} the growth properties of equation \eqref{eq:can1}
were explained proving that such equation is Liouville integrable.
Its Lagrangian, found with the method of \cite{Gubbiotti_dcov},
and the associated symplectic structure were presented.
For sake of completeness here we are going to present again such
properties.

The Lagrangian of the fourth canonical form \eqref{eq:can4} is the following:
\begin{equation}
    \begin{aligned}
        L_{4} &=
        X_{n}X_{n+1} X_{n+2}
        +X_{n}^{2}X_{n+1}+X_{n}X_{n+1}^{2}
        +\frac{X_{n}^{3}}{3}
        \\
        &+\alpha\log\left( X_{n} \right)+\beta X_{n}
        +\gamma X_{n}\left(X_{n+1}+\frac{X_{n}}{2}\right)
    \end{aligned}
    \label{eq:Lagr4}
\end{equation}

Interestingly enough, the after performing all the appropriate substitutions
in equations \eqref{eq:int1xx1} and \eqref{eq:Jint} we obtain the same
invariant.
So, from the general picture, we can produce only one invariant.
However, by direct computation we can find a second invariant:
\begin{equation}
    \begin{aligned}
        J_{4}
        &=
        - \alpha\gamma(X_{n}+X_{n-1}) 
        -\beta \gamma X_{n} X_{n-1}
        -\gamma^2X_{n} X_{n-1} (X_{n}+X_{n-1}) 
        \\
        &+\alpha(X_{n}^2+2 X_{n} X_{n-1}+X_{n} X_{n+1}+X_{n-2} X_{n-1}+X_{n-1}^2) 
        \\
        &+\beta X_{n} X_{n-1} (X_{n}+X_{n-2}+X_{n-1}+X_{n+1}) 
        \\
        &+\gamma X_{n} X_{n-1} (X_{n} X_{n-2}+2 X_{n-2} X_{n+1}+X_{n-1} X_{n+1}) 
        \\
        &+X_{n} X_{n-1} (X_{n}+X_{n-2}+X_{n-1}+X_{n+1}) \cdot
        \\
        &\phantom{+\cdot}(X_{n}^2+2 X_{n} X_{n-1}+X_{n} X_{n+1}+X_{n-2} X_{n-1}+X_{n-1}^2)
    \end{aligned}
    \label{eq:Jint4}
\end{equation}
It is easy to prove that this invariant is functionally independent from the
one obtained by performing the appropriate substitutions in \eqref{eq:int1xx1}.

Now, to complete the explicit proof of integrability of the fourth canonical
form we need to prove that the two invariants are in involution.
To this end we present the form of the sympectic structure obtained from
the Lagrangian \eqref{eq:Lagr4}.
Such Poisson structure has the following
non-zero brackets:
\begin{subequations}
    \begin{align}
        \left\{ X_{n+1},X_{n-1} \right\} &= 
        -\frac{1}{X_n},
        \\
        \left\{ X_{n+1},X_{n-2} \right\} &= 
        \frac{2 X_{n} + 2 X_{n-1}+ X_{n-2} + X_{n+1}+\gamma}{%
        X_{n}X_{n-1}}
        \\
        \left\{ X_{n},X_{n-2} \right\} &= 
        -\frac{1}{X_{n-1}},
    \end{align}
    \label{eq:J4}
\end{subequations}
Using such Poisson structure it is possible to prove that 
\emph{mutatis mudandis} the
two invariants \eqref{eq:int1xx1} and \eqref{eq:Jint4} 
are commuting.
This ends the proof of the integrability of the fourth canonical canonical form
\eqref{eq:can4}.

\paragraph{Case 5}
If $\deg g = 0$ we can write $g$ in the following way:
\begin{equation}
    g\left( \xi \right) = \kappa \neq 0.
    \label{eq:g5}
\end{equation}
In this case equation \eqref{eq:addint} reduces to the
\emph{linear} equation:
\begin{equation}
    \kappa \left(X_{n+2}+ X_{n-2}\right)
    +A_7\left( X_{n+1}+X_{n-1}\right)+\frac{(\kappa^2 +A_6 )X_{n}+A_8 \kappa+A_5}{\kappa}=0,
    \label{eq:urcan5}
\end{equation}
where for consistency we defined $X_{n}=x_{n}$.

Using the fact that $\kappa\neq0$ we make the following reparametrisation:
\begin{equation}
        A_5  =\kappa \left( \alpha \kappa - A_{8} \right),
        \quad
        A_6 = \kappa^{2}\left( \beta-1 \right)
        \quad
        A_7 = \kappa \gamma,
    \label{eq:par6}
\end{equation}
which brings equation \eqref{eq:urcan5} into:
\begin{equation}
    X_{n+2}+ X_{n-2}
    +\gamma \left( X_{n+1}+X_{n-1}\right)+\beta X_{n} + \alpha=0.
    \label{eq:can5}
\end{equation}
This is the fifth canonical form.
The fifth canonical form \eqref{eq:can5} is a degenerate case as it linear.
However, we deem it to be still interesting, as we are not just discussing
integrability, but also its relationship with Lagrangian structure.

In particular we have that the Lagrangian of the fifth canonical form \eqref{eq:can5} 
is the following:
\begin{equation}
        L_{5} =
        X_{n}X_{n+2}
        +\alpha X_{n}
        +\frac{\beta}{2}X_{n}^{2}
        +\gamma X_{n}X_{n+1}.
    \label{eq:Lagr5}
\end{equation}

Again, after performing all the appropriate substitutions
in equations \eqref{eq:int1xx1} and \eqref{eq:Jint} we obtain the same
invariant.
By direct computation we can find a second invariant:
\begin{equation}
    \begin{aligned}
        J_{5}
        &=
        \alpha(X_{n}+X_{n-2}+X_{n-1}+X_{n+1})- \alpha\gamma(X_{n}+X_{n-1})  
        \\
        &-\beta \gamma X_{n} X_{n-1}+\beta(X_{n} X_{n-2}+X_{n-1} X_{n+1}) +2\gamma X_{n+1}  X_{n-2}
        \\
        &-\gamma^2(X_{n}^2+X_{n-1}^2) +X_{n}^2+X_{n-2}^2+X_{n-1}^2+X_{n+1}^2
    \end{aligned}
    \label{eq:Jint5}
\end{equation}
It is easy to prove that this invariant is functionally independent from the
one obtained by performing the appropriate substitutions in \eqref{eq:int1xx1}.

Now, to complete the explicit proof of integrability of the fifth canonical
form we need to prove that the two invariants are in involution.
To this end we present the form of the sympectic structure obtained from
the Lagrangian \eqref{eq:Lagr5}.
Such Poisson structure has the following
non-zero, constant brackets:
\begin{equation}
        \left\{ X_{n+1},X_{n-1} \right\} = \left\{ X_{n},X_{n-2} \right\}=
        -1,
        \quad
        \left\{ X_{n+1},X_{n-2} \right\} = \gamma.
    \label{eq:J5}
\end{equation}
Using such Poisson structure it is possible to prove that 
\emph{mutatis mudandis} the
two invariants \eqref{eq:int1xx1} and \eqref{eq:Jint5} 
are commuting.
This ends the proof of the integrability of the fifth canonical canonical form
\eqref{eq:can5}.

\begin{remark}
    Following the results of Example \ref{ex:lineq} we have
    that the fifth canonical form is (up to reparametrisation) 
    the most general linear
    fourth-order difference equation admitting an \emph{autonomous}
    discrete Lagrangian.
    \label{rem:lineq}
\end{remark}

To end this section, we summarise our results in the following
theorem:

\begin{theorem} 
    If an additive Lagrangian difference equation of the form \eqref{eq:rec4addlagraut}
    admits a multi-affine invariant with respect to the variables $x_{n+1}$
    and $x_{n-2}$ of the form \eqref{eq:int1xx1}, then it is Liouville integrable .
    Moreover, using a linear transformation it can be brought into one of
    the five canonical forms given by equations \eqref{eq:can1}, \eqref{eq:can2},
    \eqref{eq:can3},\eqref{eq:can4} and \eqref{eq:can5}.
    \label{thm:intmultinlin}
\end{theorem}

\section{Continuum limits of the integrable cases}
\label{sec:contlim}

In this section we discuss the continuum limits of the
six canonical forms.
We will prove that, under appropriated scaling of the
dependent variable and of the parameters, the continuum limit are
given by either by the autonomous second member of the $P_\text{I}$ hierarchy,
the $P_\text{I}^{(2)}$ equation \cite{CresswellJoshi1999book,Kudryashov1997}:
\begin{equation}
    x\lagrangeprime{4}
    +10 x x''+\frac{r_1}{2} x''+5 \left(x'\right)^2+10 x^3
    +\frac{3}{2} r_{1} x^2+2 r_{2} x+r_{3}=0,
    \label{eq:PI2}
\end{equation}
or by the autonomous second member of the $P_\text{II}$ hierarchy,
the $P_\text{II}^{(2)}$ equation \cite{FlaschkaNewell1980}:
\begin{equation}
    x\lagrangeprime{4}-(10 x^2+r_{1}) x''+6 x^5+2r_{1} x^3 -10 x (x')^2+r_{2}
    =0.
    \label{eq:PII2}
\end{equation}
The fifth canonical form, i.e. the linear equation
\eqref{eq:can5}, is a special case.
The natural continuum limit of the fifth canonical form \eqref{eq:can5}
is the linear equation:
\begin{equation}
    x\lagrangeprime{4} + r_{1} x'' + r_{2} x + r_{3}=0.
    \label{eq:linearlim}
\end{equation}

As discussed in the general non-autonomous case in 
\cite{CresswellJoshi1999book,Kudryashov1997,FlaschkaNewell1980} the
autonomous $P_\text{I}^{(2)}$ equation \eqref{eq:PI2} and the
autonomous $P_\text{II}^{(2)}$ equation \eqref{eq:PII2} are integrable
fourth-order equations.
The linear equation \eqref{eq:linearlim} is clearly integrable.
For sake of completeness here we show their integrals and their
Lagrangian.
We note that the Lagrangian for \eqref{eq:PII2} was already presented
in \cite{Gubbiotti_dcov} using the continuum limit approach.

The autonomous $P_\text{I}^{(2)}$ equation \eqref{eq:PI2} possesses the
following first integrals
\begin{subequations}
    \begin{align}
        K_{1,\text{I}} & 
        \begin{aligned}[t]
        &= 
        x' x'''+\frac{5 x^4}{2}+\frac{ r_{1}}{2}x^3 
        +r_{2}x^2 
    +\frac{x}{16} \left[80 (x')^2+16 r_{3})\right]
    \\
        &+\frac{r_{1}}{4} (x')^2-\frac{(x'')^2}{12}
        \end{aligned}
        \label{eq:Int1PI2}
        \\
        K_{2,\text{I}} & 
        \begin{aligned}[t]
            &=
            (x''')^2
            +(20 x+r_{1}) \frac{(x'')^2}{2}
            -(60 x^2+6 x r_{1}+4 r_{2}) \frac{(x')^2}{2}
            \\
            &+(40 x^3+6 x^2 r_{1}+8 x r_{2}-4 x'^2+4 r_{3})\frac{ x''}{2}
            \\
            &-\frac{3x^{2}}{2} \left(r_{1} x^2 +4 x^3+\frac{8r_{2}}{3} x +4 r_{3}\right),
        \end{aligned}
        \label{eq:Int2PI2}
    \end{align}
    \label{eq:IntPI2}
\end{subequations}
while the autonomous $P_\text{II}^{(2)}$ equation \eqref{eq:PII2} possesses
the following first integrals:
\begin{subequations}
    \begin{align}
        K_{1,\text{II}} & = 
        x' x''' -\frac{(x'')^2}{2}-(10 x^2+r_{1}) \frac{(x')^2}{2}
        +\frac{ x}{2} (2 x^5+r_{1} x^3-r_{2}x -2 r_{3}),
        \label{eq:Int1PII2}
        \\
        K_{2,\text{II}} & 
        \begin{aligned}[t]
            &= (x''')^2-(10 x^2+r_{1}) (x'')^2
            +(x')^4+(30 x^4+6 r_{1} x^2-r_{2}) (x')^2
            \\
            &
            +\left[12 x^5+4 r_{1} x^3+4 x (x')^2-2r_{2} x -2 r_{3}\right] x''
            \\
            &+x^{3}\left[3 x^4\left(x-r_{2}\right) + 2 r_{1}x^3 -8 r_{3}\right].
        \end{aligned}
        \label{eq:Int2PII2}
    \end{align}
    \label{eq:IntPII2}
\end{subequations}
Moreover, the autonomous $P_\text{I}^{(2)}$ equation \eqref{eq:PI2} 
can be derived by the following Lagrangian:
\begin{equation}
    L_\text{I} =
    \frac{(x'')^{2}}{2} 
    + x\left( 21 x+r_{{1}}   \right) \frac{x''}{4}   
    +\frac{11x}{2}    \left(x'    \right) ^{2}
    +1/2 x\left( 5  x^{3}+r_{{1}} x^{2}+2 r_{{2}}x+2 r_{{3}} \right),
    \label{eq:LPI2}
\end{equation}
while the autonomous $P_\text{II}^{(2)}$ equation \eqref{eq:PII2} 
can be derived by the following Lagrangian:
\begin{equation}
    L_\text{II} =
    \frac{(x'')^{2}}{2}
    -x\left(\frac{5x^2}{3}+\frac{r_1}{2}\right)x''
    +x\left(x^5+\frac{r_{1}}{2}x^3+r_{2}\right).
    \label{eq:LPII2}
\end{equation}
We recall that following \cite{Fels1996} these Lagrangians
are unique up to the addition of a total derivative and multiplication
by a scalar.
Finally the linear equation \eqref{eq:linearlim} can be derived
by the following Lagrangian:
\begin{equation}
    L_\text{lin} = \frac{(x'')^{2}}{2}-\frac{r_1}{2} \left( x' \right)^2 + \frac{r_1}{2} x^2+ r_{3} x.
    \label{eq:L0lin}
\end{equation}

\begin{remark}
    We note that according to the result of \cite{Fels1996}
    the most general fourth-order linear differential equation
    admitting a Lagrangian is the following one:
    \begin{equation}
        x\lagrangeprime{4} + r_{0} x''' +r_{1} x'' +
        \frac{r_{0}}{2}
        \left(r_{1}-\frac{r_{0}^2}{4}\right)x'+r_{2}x+r_{3}=0.
        \label{eq:felslin}
    \end{equation}
    \label{eq:fels}
    The Lagrangian of the above equation is:
    \begin{equation}
        L_{F} = 
        e^{r_{0}t/2}
        \left[
            \frac{(x'')^{2}}{2}+\left(\frac{r_{0}^{2}}{8}-\frac{r_1}{2}\right) \left( x' \right)^2 
            + \frac{r_1}{2} x^2+ r_{3} x
        \right].
        \label{eq:Lfels}
    \end{equation}
    It follows from this consideration that equation \eqref{eq:linearlim}
    is the most general fourth-order linear differential equation
    admitting an \emph{autonomous} Lagrangian.
\end{remark}

\subsection{Equations reducing to the $P_\text{I}^{(2)}$ equation \eqref{eq:PI2}}
The second canonical form \eqref{eq:can2} under the following scaling:
\begin{equation}
    \begin{gathered}
        x_{n} =1 +\frac{h^{2}}{2}  x(t), \, t = nh, \,
        \alpha = -16+2 r_{1} h^2-2 r_{2} h^4,
        \\
    \beta = 30-3 r_{1} h^2+2 r_{2} h^4, \,
    \gamma = -10+\frac{r_{1}}{2} h^2+\frac{r_{3}}{4} h^6,
    \end{gathered}
    \label{eq:scal3}
\end{equation}
in the limit $h\to0$ reduces to equation \eqref{eq:PI2}.
Using the same scaling we have that the discrete Lagrangian
\eqref{eq:L3} has the following limit as $h\to 0$:
\begin{equation}
    \frac{4 L_{2}}{h^{8}} \lequiv L_\text{I} + O\left( h \right).
    \label{eq:L3clim}
\end{equation}
In the same way the invariants of the second canonical form
\eqref{eq:can2} have the following behaviour as $h\to0$:
\begin{equation}
    I_{2} = -\frac{K_{1,\text{I}}}{2} h^{8} + O(h^{9}),\,
    J_{2} = -\left(\frac{K_{1,\text{I}}}{32}+\frac{r_{1}r_{3}}{32}\right) h^{8} + O(h^{9}).
    \label{eq:I3J3lim}
\end{equation}
The two invariants collapse in a single first integral
in the continuum limit and the second one is not recovered.

\begin{remark}
    The above result on the second canonical form \eqref{eq:can2}
    clarifies that the new equation found in \cite{GJTV_class} can
    be interpreted as new autonomous discrete fourth-order Painlev\`e I equation.
    This lead us to conjecture that this equation is the fourth-order
    member of a ``non-standard'' discrete Painlev\'e I hierarchy, different from
    the one considered in \cite{CresswellJoshi1999book}.
    At the moment, no information on the existence of this hierarchy
    is available.
    \label{rem:can3}
\end{remark}

The third canonical form \eqref{eq:can3} under the following scaling:
\begin{equation}
    \begin{gathered}
        x_{n} =1 + h^{2}  x(t), \, t = nh, \,
        \alpha = -16+4 r_{1} h^2-4 r_{2} h^4,
        \\
    \beta = 56-8 r_{1} h^2, \,
    \gamma = -14+r_{1} h^2+r_{2} h^4+r_{3} h^6,
    \end{gathered}
    \label{eq:scal4}
\end{equation}
in the limit $h\to0$ reduces to equation \eqref{eq:PI2}.
Using the same scaling we have that the discrete Lagrangian
\eqref{eq:L4} has the following limit as $h\to 0$:
\begin{equation}
    \frac{L_{3}}{h^{8}} \lequiv L_\text{I} + O\left( h \right).
    \label{eq:L4clim}
\end{equation}
In the same way the invariants of the third canonical form
\eqref{eq:can3} have the following behaviour as $h\to0$:
\begin{equation}
    I_{3} = -8 K_{1,\text{I}} h^{8} + O(h^{9}),\,
    J_{3} = -\left(136 K_{1,\text{I}}+\frac{6r_{1}r_{3}+5r_{2}^{2}}{17}\right) h^{8} + O(h^{9}).
    \label{eq:I4J4lim}
\end{equation}
The two invariants collapse in a single first integral
in the continuum limit and the second one is not recovered.

\begin{remark}
    It was noted in section \ref{sec:int} that the third
    canonical form \eqref{eq:can3} is related to the first
    one \eqref{eq:can1}. 
    Since the first canonical form \eqref{eq:can1} is an autonomous 
    $\dP_\text{II}^{(2)}$
    it would be natural to identify also the third canonical form
    \eqref{eq:can3} with the autonomous fourth-order member of the Painlev\'e II
    hierarchy.
    However, the simplest continuum limit of the third
    canonical form \eqref{eq:can3} is the autonomous fourth-order 
    member of the Painlev\'e I hierarchy.
    No continuum limit of this equation to the autonomous 
    fourth-order member of the Painlev\'e II hierarchy is at present known.
    Finally, as in remark \ref{rem:can3}
    it is not known if the third canonical form \eqref{eq:can3} is the fourth-order
    member of a ``non-standard'' discrete Painlev\'e I hierarchy, different from
    the one considered in \cite{CresswellJoshi1999book}.
    \label{rem:dPII2toPI2}
\end{remark}

The fourth canonical form \eqref{eq:can4} under the following scaling:
\begin{equation}
    \begin{gathered}
        x_{n} =1 + h^{2}  x(t), \, t = nh, \,
        \alpha = -10+\frac{3 r_{1}}{2} h^2-r_{2}h^4 ,
        \\
    \beta = 30-3r_{1} h^2 , \,
    \gamma =-10+\frac{r_{1}}{2} h^2+\frac{r_{2}}{3} h^4+ \frac{r_{3}}{3} h^6,
    \end{gathered}
    \label{eq:scal5}
\end{equation}
in the limit $h\to0$ reduces to equation \eqref{eq:PI2}.
Using the same scaling we have that the discrete Lagrangian
\eqref{eq:L5} has the following limit as $h\to 0$:
\begin{equation}
    \frac{L_{4}}{h^{8}} \lequiv L_\text{I} + O\left( h \right).
    \label{eq:L5clim}
\end{equation}
In the same way the invariants of the fourth canonical form
\eqref{eq:can4} have the following behaviour as $h\to0$:
\begin{equation}
    I_{4} = 2 K_{1,\text{I}} h^{8} + O(h^{9}),\,
    J_{4} = \left(32 K_{1,\text{I}}+\frac{6r_{1}r_{3}}{24}+\frac{r_{2}^{2}}{36}\right) h^{8} + O(h^{9}).
    \label{eq:I5J5lim}
\end{equation}
The two invariants collapse in a single first integral
in the continuum limit and the second one is not recovered.
This continuum limit was first discussed in \cite{CresswellJoshi1999book}.

\subsection{Equations reducing to the $P_\text{II}^{(2)}$ equation \eqref{eq:PII2}}

The first canonical form \eqref{eq:can1} under the following scaling:
\begin{equation}
        x_{n} = h  x(t), \, t = nh, \,
        \alpha = 6+2 r_{1} h^2 ,\,
    \beta = r_{2} h^5 , \,
    \gamma =  4+r_{1} h^2, 
    \label{eq:scal2}
\end{equation}
in the limit $h\to0$ reduces to equation \eqref{eq:PII2}.
Using the same scaling we have that the discrete Lagrangian
\eqref{eq:L4} has the following limit as $h\to 0$:
\begin{equation}
    \frac{L_{1}}{h^{6}} \lequiv -L_\text{II} + O\left( h \right).
    \label{eq:L2clim}
\end{equation}
In the same way the invariants of the first canonical form
\eqref{eq:can1} have the following behaviour as $h\to0$:
\begin{equation}
    I_{1} = -2 \alpha x_{1}^{6} K_{1,\text{II}} h^{6} + O(h^{7}),\,
    J_{1} = \alpha^{4}x_{1}^{8} K_{1,\text{II}}h^{6} + O(h^{7}).
    \label{eq:I2J2lim}
\end{equation}
The two invariants collapse in a single first integral
in the continuum limit and the second one is not recovered.
This continuum limit was first discussed in \cite{CresswellJoshi1999}.

\subsection{Equation reducing to equation \eqref{eq:linearlim}}
The fifth canonical form \eqref{eq:can5} under the following scaling:
\begin{equation}
    x_{n} = x(t), \, t = nh, \,
    \alpha = r_{3} h^{4}, \,
    \beta = 6 - 2r_{1}h^2+r_{2} h^4, \,
    \gamma = -4 + r_{1}h^{2},
    \label{eq:scal6}
\end{equation}
in the limit $h\to0$ reduces to equation \eqref{eq:linearlim}.
Using the same scaling we have that the discrete Lagrangian
\eqref{eq:Llintransf} has the following limit as $h\to 0$:
\begin{equation}
    \frac{L_{5}}{h^{4}} \lequiv L_\text{lin} + O\left( h \right).
    \label{eq:L6clim}
\end{equation}

Since equations \eqref{eq:can5} and \eqref{eq:linearlim} are linear
instead of discussing the relationship between the invariants
we discuss the relationship between the explicit solutions.
The explicit solutions of equation \eqref{eq:can5} is obtained
as linear combination of the base solutions $X_{n,i} = q_{i}^{n}$, where
$q_{i}$ are the four roots of the characteristic polynomial:
\begin{equation}
    q^{4} + \gamma q^{3} + \beta q^{2} +\gamma q + 1 =0,
    \label{eq:charcan6}
\end{equation}
plus a particular solution of the inhomogeneous equation.
In the same way the general solution of \eqref{eq:linearlim} 
is obtained through as linear combination of 
the base solutions $x_{n,i} = e^{\mu_{i} t}$, where
$\mu_{i}$ are the four roots of the characteristic polynomial:
\begin{equation}
    \mu^{4} + r_{1}\mu^{2} + r_{2} =0,
    \label{eq:charlin}
\end{equation}
plus a particular solution of the inhomogeneous equation.
The solutions of equation \eqref{eq:charlin} are obtained from
the solutions of equation \eqref{eq:charcan6} using the scaling given
in formula \eqref{eq:scal6} and
\begin{equation}
    q = 1 + \mu h,
    \label{eq:qtomu}
\end{equation}
in the limit where $h\to0$.
Indeed, using formula \eqref{eq:qtomu} into \eqref{eq:charcan6}
we obtain:
\begin{equation}
    \left( \mu^{4} + r_{1}\mu^{2} + r_{2} \right)h^{4}
    + O\left( h^{5} \right) = 0.
    \label{eq:char5tolin}
\end{equation}
Finally, using $t = n h$ the base solutions are such that:
\begin{equation}
    X_{n,i} = \left( 1+\mu_{i} h \right)^{t/h}
    = e^{\mu_{i} t} + O(h)=x_{i}(t)+O(h).
    \label{eq:basesol}
\end{equation}
An analogous result holds for the particular solution if
we write down its expression using the method of variation
of constants \cite{Elaydi2005}.

\section{Conclusions}
\label{sec:concl}

In this paper we discussed the conditions for an additive
fourth-order difference equation \eqref{eq:rec4add} to admit
a Lagrangian.
Our main result, stated in Theorem \ref{thm:structure}, tells
us that there exists a family of such equations depending
on two arbitrary functions, one of a single variable $g=g\left( \xi \right)$,
and one of two variables $V=V\left( \xi,\eta \right)$, 
and on an arbitrary constant $\lambda$.
As evidenced in Corollary \ref{cor:structureaut} the
Lagrangian is autonomous if and only if $\lambda=1$.

Additive difference equations can be considered also in higher
dimension.
Indeed, an additive $2k$th-order difference equation is a
difference equation of the following form:
\begin{equation}
    x_{n+k} = f\bigl( \vec{x}_{n}^{(-k+1,k-1)} \bigr) x_{n-k}
    +h\bigl( \vec{x}_{n}^{(-k+1,k-1)} \bigr),
    \label{eq:addgen}
\end{equation}
where we defined
\begin{equation}
    \vec{x}_{n}^{(m,l)} = \left( x_{n+m},\dots,x_{n+k} \right),
    \quad
    l\leq m.
    \label{eq:xnml}
\end{equation}
The result of this paper stimulate to consider the following
conjecture regarding difference equations like \eqref{eq:addgen}:
\begin{conjecture*}
    An additive $2k$th-order difference equation is variational
    \emph{if and only if} it can be derived from the following Lagrangian:
    \begin{equation}
        L_{n}^{(k)} =
        \lambda^{-n}
        \left[  
            f\bigl( \vec{x}_{n}^{(1,k-1)} \bigr) x_{n}x_{n+k} 
            +V\bigl( \vec{x}_{n}^{(0,k-1)} \bigr)
        \right].
        \label{eq:lagrconj}
    \end{equation}
\end{conjecture*}
The study of this conjecture will be subject of further studies.
A starting point for these studies are the known hierarchies of
discrete equations, e.g. those presented in 
\cite{CresswellJoshi1999book,CresswellJoshi1999}.

Moreover, in this paper 
to better underline the power of the Lagrangian approach we
produced a list of integrable equations
with autonomous Lagrangian using an ansatz on the 
shape of one invariant.
Interestingly enough, equations of the said list naturally 
possess a second invariant without imposing any additional
conditions.
We showed that it is possible to reduce these equations
to five canonical forms, which we related to known examples from
\cite{CresswellJoshi1999book,CresswellJoshi1999,GJTV_class,JoshiViallet2017}.
We remark that in the cited papers, the same equations were derived 
or studied with different approaches.

Finally, we computed the continuum limits of the canonical forms.
This allowed us to identify equation \eqref{eq:can2},
an equation recently introduced in \cite{GJTV_class},
with a new discrete $P_\text{I}^{\left( 2 \right)}$ equation.
Moreover, the continuum limits showed that 
equation \eqref{eq:can3}, which is related to the discrete 
$P_\text{II}^{\left( 2 \right)}$ equation \eqref{eq:can1},
is actually a discretisation of the $P_\text{I}^{\left( 2 \right)}$
equation.
In the same way we proved, following the example given in \cite{Gubbiotti_dcov},
that variational structures are preserved upon continuum limit,
while invariants are not.
A resuming table of the integrable case, and their continuum limits
can be found in Table \ref{tab:integrable}.

\begin{table}
    \centering
    \begin{tabular}{ccccc}
        \toprule
        Canonical form & Equation & Roots of $g\left( \xi \right)$  & Continuum limit & Introduced
        \\
        \midrule
        1st & \eqref{eq:can1} & -1,1 & autonomous $P_\text{II}^{(2)}$ & 
        \cite{CresswellJoshi1999,JoshiViallet2017}
        \\
        2nd & \eqref{eq:can2} & 0,0 & autonomous $P_\text{I}^{(2)}$ & \cite{GJTV_class}
        \\
        3rd & \eqref{eq:can3} & $-\imath,\imath$ & autonomous $P_\text{I}^{(2)}$ &
        --
        \\
        4th & \eqref{eq:can4} & 0 & autonomous $P_\text{I}^{(2)}$ & 
        \cite{CresswellJoshi1999book,JoshiViallet2017}
        \\
        5th & \eqref{eq:can5} & -- & equation \eqref{eq:linearlim} & --
        \\
        \bottomrule
    \end{tabular}
    \caption{Resuming table of the integrable canonical forms.}
    \label{tab:integrable}
\end{table}

Except that in Section \ref{sec:main} we did not dealt with
autonomous Lagrangians, but now we would like to give a
interpretation of their appearance, based on the analogy
with the continuum systems.
From the results of Example \ref{ex:lineq} and from 
the continuum limit \eqref{eq:linearlim} of the fifth canonical 
from \eqref{eq:can5} we infer that non-autonomous Lagrangians
are linked to some form of \emph{dissipation}.
We propose this analogy for two main reasons.
First because in the continuum limit \eqref{eq:linearlim} 
odd-order derivates are absent.
Odd-order derivates are naturally related to
dissipation in continuous systems.
Second, we can prove that the additive variational fourth-order
equations with non-autonomous Lagrangians are not 
\emph{measure preserving}, but they either shrink or expand the
volume of the phase space.
Indeed if we compute the Jacobian determinant of the most 
general variational additive fourth-order difference equation 
\eqref{eq:rec4addlagr} we obtain:
\begin{equation}
    J_{n} = \lambda^{2} \frac{g\left( x_{n-1} \right)}{g\left( x_{n+1} \right)}.    
    \label{eq:Jacgen}
\end{equation}
This implies that the volume element is given by:
\begin{equation}
    V_{n} = g\left( x_{n} \right)g\left( x_{n-1} \right)
    \ud x_{n+1}\wedge\ud x_{n}\wedge
    \ud x_{n-1}\wedge\ud x_{n-2}
    \label{eq:voln}
\end{equation}
and evolves according to:
\begin{equation}
    V_{n+1} = \lambda^{2} V_{n},
    \label{eq:evolvol}
\end{equation}
that is $V_{n} = \lambda^{2n} V_{0}$.
We obtain that if $\abs{\lambda}>1$ the volume of the phase space
is increasing, while if $0<\abs{\lambda}<1$ the volume of the phase
space is decreasing.
This is another usual feature of continuous dissipative
equations.
If and only if $\lambda=1$ we have the conservation of the volume 
as required by the Hamiltonian approach. 
For the above reasons we say the autonomous Lagrangian
case is \emph{conservative}, while the non-autonomous one
is \emph{dissipative}.
This behaviour is displayed graphically in Figure \ref{fig:dissipation}
in the case of the first canonical form \eqref{eq:can1} and its asymmetric
version obtained from the discrete Lagragian $L_{n,2}=\lambda^{-n}L_{2}$ 
with a given $\lambda\in\left( -1,1 \right)$.
See remark \ref{rem:nonautL}.

\begin{figure}[t]
    \centering
    \includegraphics{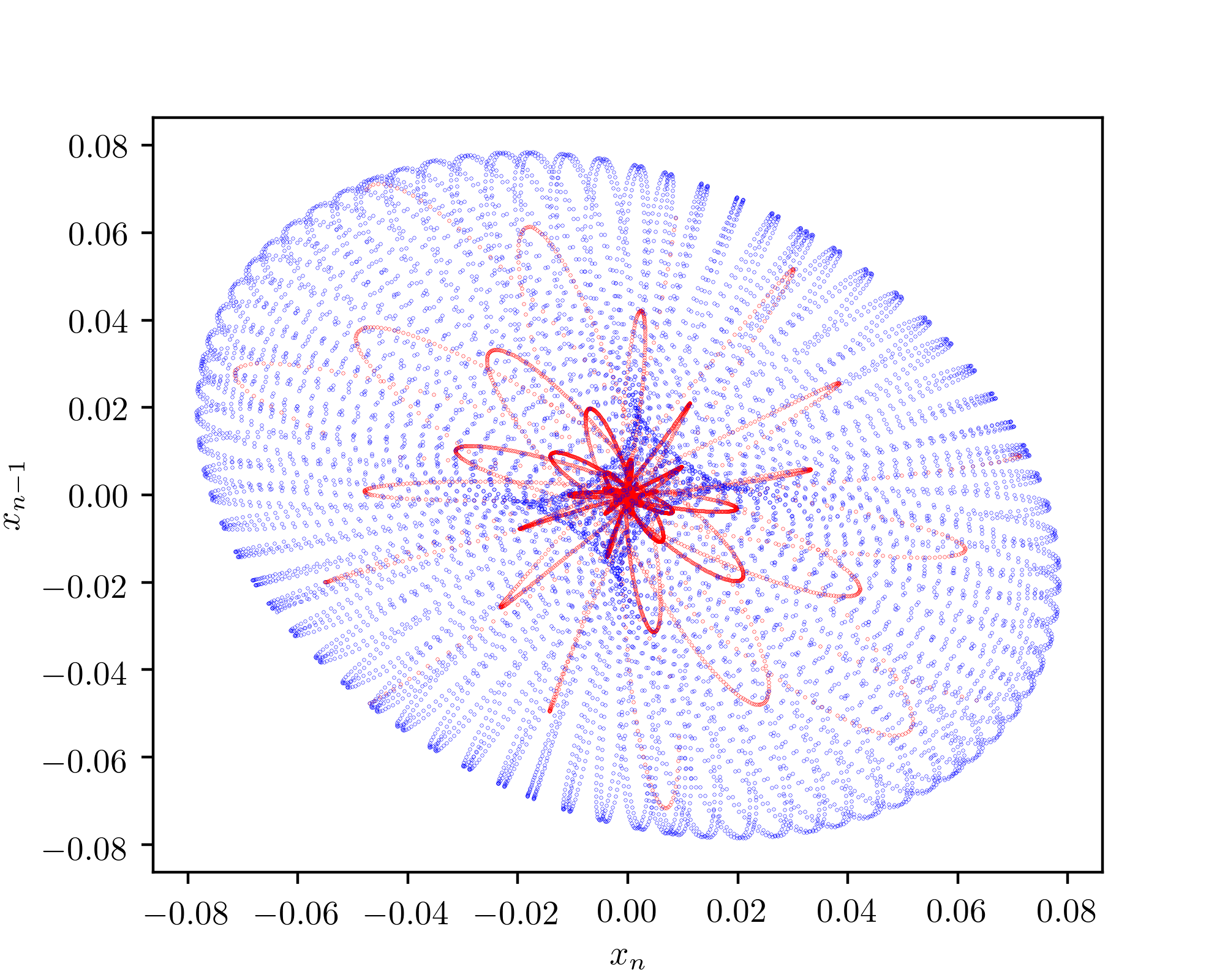}
    \caption{In blue a trajectory of equation \eqref{eq:can1} with $A_{5}=2$, $A_{6}=0$,
    $A_{7}=-1$ and initial conditions $x_{i}\sim 10^{-2}$.
    In red a trajectory of the equation obtained from $L_{n,2}=\lambda^{-n}L_{2}$,
with $\lambda=0.999$, same parameters and same initial conditions.
While the trajectory of \eqref{eq:can1} oscillated around the fixed point in the origin,
the asymmetric trajectory collapse into it as $n\to\infty$. Trajectories are computed
using $10^{4}$ iterations.}
    \label{fig:dissipation}
\end{figure}

It is well known that dissipative systems are not
integrable in the usual Liouville sense, as they fail to
preserve the measure of the phase space.
On the other side, in the continuous setting it is also known that 
some dissipative systems admit \emph{time-dependent} 
first integrals \cite{Caldirola1941,Kanai1948}.
Up to our knowledge such possibility has never been explored
in the discrete setting, so this raises the following
question:

\begin{problem*}
    Do \emph{non-trivial} variational discrete systems 
    \emph{admitting $n$-dependent invariants} exist?
\end{problem*}

Here, by non-trivial we mean a system for which it is not
possible to write down the general solution and invert it
with respect to the initial conditions in order to get the
$n$-dependent invariants.
This restriction is important to rule out linear system, for
which this procedure is always possible.
This problem might be interesting from the point of view of
applications as in several real cases one might need to take into
account dissipative effects caused e.g. by friction.
We are planning to address to this problem in a future study.

Other application of the result of this paper can arise
in the field of \emph{geometric integration theory} 
\cite{BuddIserles1999,BuddPiggot2003,KrantzParks2008}.
Geometric integration theory is a branch of numerical analysis
which deals in preserving properties when discretising a
continuous system.
The variational structure might such a property.
To give a very simple example consider the following
Lagrangian:
\begin{equation}
    L = \frac{\left( x'' \right)^{2}}{2}
    -\alpha\frac{\left( x' \right)^{4}}{12} + \frac{\omega^{2}}{2}x^{2}-\beta x
    \label{eq:Lbeam}
\end{equation}
and its Euler--Lagrange equation:
\begin{equation}
    x\lagrangeprime{4} + \alpha \left( x' \right)^{2} x''+ \omega^{2} x = \beta.
    \label{eq:beam}
\end{equation}
A trivial discretisation of equation \eqref{eq:beam} is obtained
by replacing the derivatives with the discrete derivatives:
\begin{equation}
    x' \to \delta_{n} x_{n} = \frac{x_{n}-x_{n-1}}{h}.
    \label{eq:triv}
\end{equation}
The resulting discrete equation is (up to translation in $n$):
\begin{equation}
    \begin{aligned}
        &\frac{x_{n+2}-4 x_{n+1}+6 x_{n}-4 x_{n-1}+x_{n-2}}{h^4}
        \\
        +\alpha&\frac{(x_{n-1}-x_{n-2})^2 (x_{n}-2 x_{n-1}+x_{n-2})}{h^4}
        +\omega^2 x_{n-2}=\beta.
    \end{aligned}
    \label{eq:trivbeam}
\end{equation}
This equation is nor invertible nor variational.
On the other hand,  there exist infinitely many variational
discretisation of equation \eqref{eq:triv} with the following
hypotheses:
\begin{itemize}
    \item $\lambda=1$,
    \item the function $g$ is a constant,
    \item the functions $M$ and $N$ in \eqref{eq:dencleared3}
        are third order polynomials in their variables,
    \item the coefficients of $M$ and $N$ are second order polynomials
        in $h$.
\end{itemize}
An example is the following one:
\begin{equation}
    \begin{aligned}
    x_{n+2} 
    &-4 x_{n+1} 
    +\left(6-\omega^{2}h^{4}\right)x_{n} -4x_{n-1}+ x_{n-2}
    \\
    &+\frac{\alpha}{3}\left( x_{n+1}+x_{n-1}-2x_{n} \right)\times
    \\
    &\phantom{\times}
    \left( x_{n+1}^{2}+x_{n}^{2}+x_{n-1}^{2}-x_{n+1}x_{n}-x_{n+1}x_{n-1}-x_{n}x_{n-1} \right)
    =h^{4}\beta.
    \end{aligned}
    \label{eq:discrtriv}
\end{equation}
This discretisation is variational by construction.
We argue that this kind of discretisation, even in
the non-integrable case might be convenient from
a numerical point of view.
This topic will be subject of future studies.

Finally, additive fourth-order difference equations are not the only
possible generalisation of second-order equations.
For instance in \cite{CapelSahadevan2001} several integrable
equations of \emph{multipliticative} form were derived:
\begin{equation}
    x_{n+2}x_{n-2} = F\left(x_{n+1},x_{n},x_{n-1}  \right).
    \label{eq:mult}
\end{equation}
In an upcoming paper we are addressing the problem of
giving necessary and sufficient conditions on the existence
of a variational structure of such equations and study their integrability
properties.

\section*{Acknowledgements}

GG would like to thank Prof. N. Joshi for the helpful discussions
during the preparation of this paper.

GG is supported by the Australian Research Council through 
Nalini Joshi's Australian Laureate Fellowship grant FL120100094.

\printbibliography


\end{document}